\documentclass[10pt,conference]{IEEEtran}
\usepackage{graphicx}
\usepackage{hyperref}
\usepackage{algorithmic}
\usepackage[ruled,vlined,linesnumbered]{algorithm2e}
\usepackage{epsfig}
\usepackage{subfig}
\usepackage{balance}  % for  \balance command ON LAST PAGE  (only there!)
\usepackage[table]{xcolor}
\usepackage{color,soul,xspace}
\usepackage{amssymb, amsmath, amsthm, amsfonts, mathrsfs }
\newcommand{\todo}[1]{\sethlcolor{red}\hl{TODO: #1}}
\newcommand{\eat}[1]{}
\DeclareMathOperator*{\argmax}{arg\,max}
\newtheorem{thm}{Theorem}
\newtheorem{cor}[thm]{Corollary}
\newtheorem{lemma}{Lemma}%[section]
\newtheorem{problem}{Problem}
\begin{document}
%\title{Maximizing Coverage Centrality via Network Design: Extended Version}

%\title{Group Centrality Maximization via Network Design}
%\begin{document}
\title{Maximizing Coverage Centrality via Network Design: Extended Version}
%\titlenote{Produces the permission block, and
 % copyright information}
%\subtitle{Extended Abstract}
%\subtitlenote{The full version of the author's guide is available as
  %\texttt{acmart.pdf} document}

%\author{Sourav Medya$^1$, Arlei Silva$^1$, Ambuj Singh$^1$, Prithwish Basu$^2$, Ananthram Swami$^3$}
%\authornote{Dr.~Trovato insisted his name be first.}
%\affiliation{%
%  \institution{$^1$Computer Science Department, University of California, Santa Barbara, CA, USA}
 % \streetaddress{$^2$Raytheon BBN Technologies, Cambridge, MA, USA, $^3$Army Research Laboratory, %Adelphi, MD, USA}
%\email{medya,arlei,ambuj@cs.ucsb.edu, pbasu@bbn.com, ananthram.swami.civ@mail.mil}

\author{
    \IEEEauthorblockN{Sourav Medya\IEEEauthorrefmark{1}, Arlei Silva\IEEEauthorrefmark{1}, Ambuj Singh\IEEEauthorrefmark{1}, Prithwish Basu\IEEEauthorrefmark{2}, Ananthram Swami\IEEEauthorrefmark{3}}
    \IEEEauthorblockA{\IEEEauthorrefmark{1}Computer Science Department, University of California, Santa Barbara, CA, USA, \{medya, arlei, ambuj\}@cs.ucsb.edu}
    \IEEEauthorblockA{\IEEEauthorrefmark{2}Raytheon BBN Technologies, Cambridge, MA, USA, pbasu@bbn.com}
    \IEEEauthorblockA{\IEEEauthorrefmark{3}Army Research Laboratory, Adelphi, MD, USA, ananthram.swami.civ@mail.mil}
}

\maketitle
\begin{abstract}

%Re-designing networks is one of the popular and possible ways to improve properties of a network. Improving network properties such as centralities of vertices improve the information flow  through them. Besides information flow, the increase in centrality makes the vertices ``closer" to all other nodes in terms of spread of influence, shortest paths and closeness centrality. 
%In this paper, we aim to increase a fundamental centrality measure of a group of nodes: coverage centrality. Coverage centrality of a set of nodes $X$ is defined as the distinct number of shortest paths go through nodes in $X$. Our objective is, given $k$ as budget, a candidate set of edges $\Gamma$, and a set of nodes $X$, how to find the $k$ edges from $\Gamma$ and add  between the nodes in set $X$ and others to maximize the coverage centrality of $X$? The challenges involve 1) the large search space of candidate edges for optimal $k$ edges and 2) the high computational cost of evaluating the large number of combinations between the candidate edges after addition. Though the problem is NP-hard, we propose a simply greedy method with constant time approximation. As the method is inefficient on large networks, we further develop an efficient randomized algorithm with a probabilistic approximation guarantee. Our experiments show that the proposed method is applicable for real world networks. 

Network centrality plays an important role in many applications. Central nodes in social networks can be influential, driving opinions and spreading news or rumors. %In physical networks that model water distribution, power grids and road systems, central nodes can monitor relevant phenomena, such as contamination, outages, and congestion. 
In hyperlinked environments, such as the Web, where users navigate via clicks, central content receives high traffic, becoming targets for advertising campaigns.
While there is an extensive amount of work on centrality measures and their efficient computation, controlling nodes' centrality via network updates is a more recent and challenging task. Performing minimal modifications to a network to achieve a desired property falls under the umbrella of network design problems. This paper is focused on improving the group (coverage and betweenness) centrality of a set of nodes, which is a function of the number of shortest paths passing through the set, by adding edges to the network. We introduce several variations of the problem, showing that they are NP-hard as well as APX-hard.  Moreover, we present a greedy algorithm, and even faster sampling algorithms, for group centrality maximization with theoretical quality guarantees under a restricted setting and good empirical results in general for several real datasets.

%We prove inapproximability results for different variations of the problem and propose a greedy algorithm for maximizing group centrality. To ensure scalability to large networks, we also design an efficient sampling algorithm for the problem. 
%In addition to providing an extensive empirical evaluation of our algorithms, we also show that, under some realistic constraints, the proposed solutions achieve optimal approximation for group centrality maximization.
%We further study a constrained yet realistic version of the problem. Besides showing APX-hardness for this restricted problem, we prove our proposed algorithm achieves nearly optimal approximation guarantee. 
%We evaluate our approach on several real-world graphs showing that it outperforms the best baseline solution in terms of quality ---coverage centrality achieved--- by up to $5$ times.

%We show that the problem is NP-hard and then propose a polynomial-time constant-factor approximation algorithm under some mild assumptions. To ensure applicability to large networks, we design a novel sampling algorithm for centrality optimization with probabilistic approximation guarantees. Results show that our approach achieves up to $5$ times higher coverage centrality than the baseline solutions on many real-world networks. Moreover, our randomized algorithm scales to million-node networks.

%\todo{Numbers}

\end{abstract}

\section{Introduction}

% A brief intro- motivate netwrok design
\textit{Network design} is a recent area of study focused on modifying or redesigning a network in order to achieve a desired property \cite{gupta2011approximation,zhu2004power}. As networks become a popular framework for modeling complex systems (e.g. VLSI, transportation, communication, society), network design provides key controlling capabilities over these systems, specially when resources are constrained. Existing work has investigated the optimization of global network properties, such as  minimum spanning tree \cite{krumke1998}, shortest-path distances \cite{lin2015,dilkina2011,meyerson2009}, diameter \cite{demaine2010}, and information diffusion-related metrics \cite{Khalil2014,Tong2012GML} via a few local (e.g. vertex, edge-level) upgrades.
%Modifying or redesigning a network structure is a way to achieve better properties of the same. Some structural properties are very important to solve many real world problems, so as a result, improving these properties become inherently important. Network design problems have been studied to optimize different structural properties of the network. These design problems have a large number of applications ~\cite{gupta2011approximation,o1994hub,lin2015,zhu2004power, Khalil2014, saha2015approximation} in different types of networks (communication, information, social etc.).
% challenges - motivate centrality control
% In these problems, local changes such as adding links, nodes or modifying their weights are made to an existing network as a means to improve its global properties~\cite{chaoji2012recommendations, Tong2012GML, dilkina2011, meyerson2009, lin2015,Khalil2014}. 
Due to the large scale of real networks, computing a global network property becomes time-intensive. For instance, computing all-pair shortest paths in large networks is prohibitive. As a consequence, %the improvement of metrics involving these global properties become very challenging
design problems are inherently challenging. %The other challenge comes from making an optimal choice from the candidate set of allowed modifications due to budget constraint. 
Moreover, because of the combinatorial nature of these local modifications, network design problems are often NP-hard, and thus, require the development of efficient approximation algorithms.

\iffalse

\begin{figure}[t]
    \centering
    \vspace{-5mm}
    \subfloat[Initial graph]{
        \includegraphics[keepaspectratio, width=.2\textwidth]{Figures/Example/netscience.png}
       \label{fig:ex11}
    }
    \subfloat[Initial graph]{
        \includegraphics[keepaspectratio, width=.1\textwidth]{Figures/Example/example2.pdf}
       \label{fig:ex11}
    }
    \subfloat[Modified graph]{
       \includegraphics[keepaspectratio, width=.1\textwidth]{Figures/Example/example3.pdf}
       \label{fig:ex12}
    }
 
     \caption{ \small Netscience (co-authorship in network science) data: top $5$ newly added edges (green) found by our method. The larger two nodes (M. Newman and A. Barabasi) are the target nodes. Their covergae centrality gets improved by $23K$. Illustrative example of the Coverage Centrality Optimization problem. Here we want to optimize the centrality of $\{d,f\}$ with a budget of one edge from the candidates $\{(d,a), (d,b), (f,b)\}$. 
     %(a) represents the initial graph and (b) represents the modified graph with additional edge $(d,a)$. In (a) and (b), the coverage centrality of $\{d,f\}$ is $0$ and $6$ respectively. 
     The coverage centrality of $\{d,f\}$ is $0$ in the initial graph (a) and $3$ in the modified graph (b). Node $d$ belongs to the shortest paths between $(a,e)$, $(a,c)$ and $(a,f)$ in (b).\label{fig:example1}}
     \vspace{-4mm}
\end{figure}
\fi

%betweenness vs coverage centrality
We focus on a novel network design problem, which is improving the \textit{group centrality}. Given a node $v$, its coverage centrality is the number of distinct node pairs for which a shortest path passes through $v$, whereas its betweenness centrality is the fraction of shortest paths between any distinct pair of nodes passing through $v$ \cite{yoshida2014}. The centrality of a group $X$ is a function of the shortest paths that go through members of $X$ \cite{yoshida2014}. \emph{Our goal is to maximize group centrality, for a target group of nodes, via a small number of edge additions.}
%from a candidate edge set}. 

As an application scenario, consider an online advertising service where advertisers can place links on a set of pages depending on user context information (see Figure \ref{fig:motivation}). For instance, users navigating from travel to car related pages are likely to be interested in car rentals. Thus, the ad service can display links in a subset of pages in order to increase the number of shortest paths from travel related web-pages to car related ones via a set of pages owned by a given car rental company. The idea is to boost the traffic to the car rental pages while users browse the Web, assuming that clicks will often follow shortest paths. Once the user arrives at an advertiser's page, the car rental company can offer targeted information to support her browsing through the automobile related content (e.g. highlighting car models often rented in a given tourist location). This problem is equivalent to optimizing the group centrality of the advertiser's pages---for a selected set of node pairs---by adding few edges from a candidate set.
%a major airline, which might operate thousands of flights with hundreds of destinations. In order to optimize its operations, the airline might want to shift its flights towards a few airports with lower delays and/or reduced costs while minimizing the disruptions over its current settings. Another concern is to keep 
%the trips short, as it reduces fuel consumption 
%the number of flight segments small, as means to reduce fuel consumption and improve passengers' comfort. One way to achieve this goal is to create a few new routes %connecting the airport network 
%to increase the coverage centrality of the targeted airports. 

Another application scenario is a professional network, such as \textit{LinkedIn}, where the centrality of some users (e.g. employees of a given company) might be increased via connection recommendations/advertising. %In particular, shortest paths have important semantics in these networks, as they show how professionals can reach each other and get referred for professional purposes via existing connections.
%Consider an organization's social network N that is used for informal interactions, information exchange, and at the same time appraising skills of individuals. Embedded in this network are also members of the organization's executive committee for assigning responsibilities and collecting feedback. This group of individuals should be centrally located in network N. However, N is a dynamic network with new ties being created, old ones removed, employees leaving and joining. In the case of these changes, we need to ensure that the executive committee is centrally located. One way to achieve this is by adding new edges to the committee members.
%We will give a few more examples where increasing coverage centrality becomes important. 
In military settings, where networks might include adversarial elements, inducing the flow of information towards key agents can enhance communication and decision making \cite{perumal2013}.
%Moreover, while social influence and information propagation are traditionally modelled as a diffusion process (i.e., using all possible paths)~\cite{kempe2003maximizing}, multiple recent approaches focus on the most probable (shortest) paths in order to allow scalable solutions~\cite{kimura2006tractable,chen2010scalable}. Thus, to achieve better information propagation through a target group of nodes of interest, one needs to improve their shortest path based centrality. %Improving coverage centrality serves the purpose.
Moreover, multiple recent approaches focus on the most probable (shortest) paths to allow scalable solutions for social influence and information propagation~\cite{kimura2006tractable,chen2010scalable}. Thus, to achieve better information propagation through a target group of nodes of interest, one might improve their shortest path based centrality. %Improving coverage centrality serves the purpose. 
%Generally speaking, coverage centrality is a fundamental metric to describe how well positioned nodes are in the network. In our experiments (see Section \ref{sec:exp}), we show that increasing coverage centrality has a positive impact on other relevant metrics, such as ``closeness" and influence \cite{kempe2003maximizing}. %In Fig.~\ref{fig:motivation}, we show the top $5$ newly added edges found by our method in \textit{Netscience}\footnote{The dataset is available here: \url{http://www-personal.umich.edu/~mejn/netdata/}} (co-authorship in network science) data. The larger two nodes (M. Newman and A. Barabasi) are the target nodes. Their coverage centrality increases by approximately $23K$. Interestingly, our algorithm selects the well positioned nodes in the network with reasonable degree to add the edges. 

From a theoretical standpoint, for any objective function of interest, we can define a \textit{search} and a corresponding \textit{design} problem. In this paper, we show that, different from its search version \cite{yoshida2014}, group centrality maximization cannot be approximated by a simple greedy algorithm. Furthermore, we study several variations of the problem and show that, under two realistic constraints, the problem does allow a constant factor greedy approximation. In fact, we are able to prove that our approximation for the constrained problem is \textit{optimal}, in the sense that the best algorithm cannot achieve a better approximation than the one obtained by our approach. In order to scale our greedy solution to large datasets, we also propose efficient sampling schemes, with approximation guarantees, for group centrality maximization. 

\begin{figure}[t]
    \centering
    \vspace{-4mm}
    \includegraphics[keepaspectratio, width=.35\textwidth]{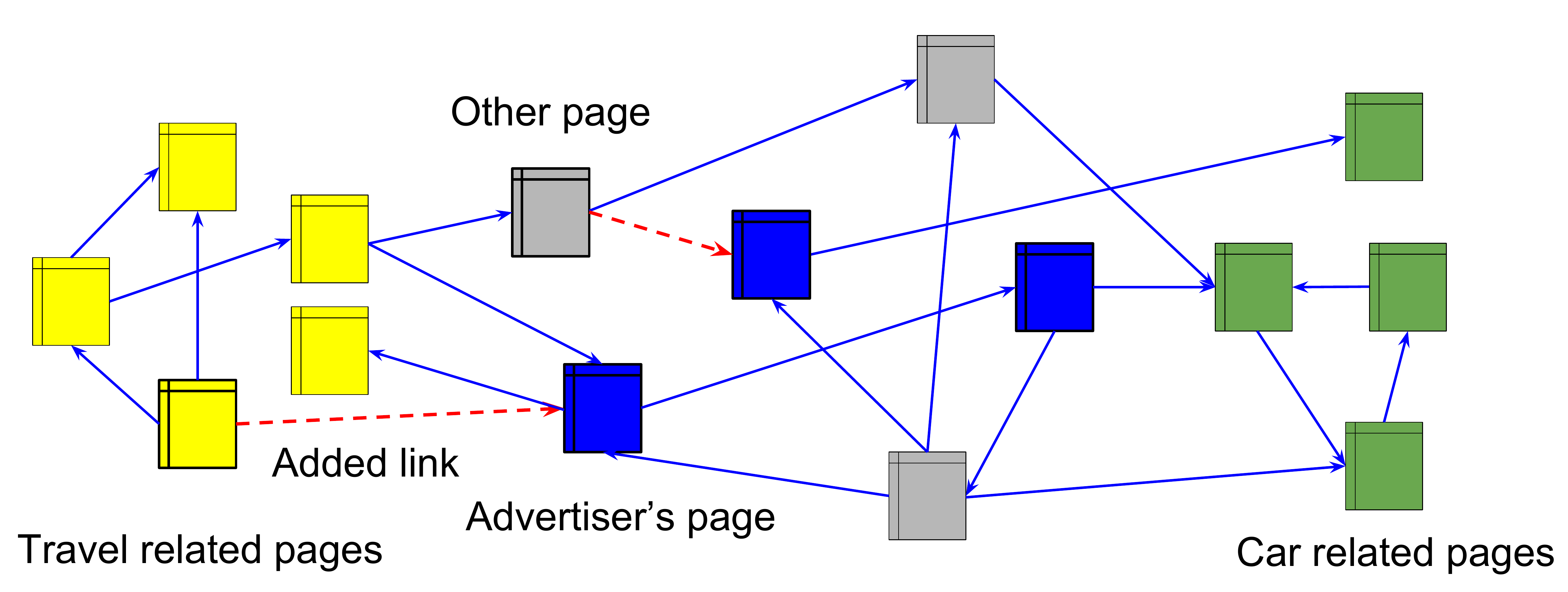}
%      \vspace{-7mm}
 
    % \caption{ \small Netscience (co-authorship in network science) data: top $5$ newly added edges (green) found by our method. The larger two nodes (M. Newman and A. Barabasi) are the target nodes. Their covergae centrality gets improved by $23K$.\label{fig:motivation}}
     %\vspace{-1mm}
      \caption{Application scenario: Adding links to the Web graph in order to increase user traffic from travel to car related pages via a set of car rental (advertiser's) pages.\label{fig:motivation}}
     
    % \vspace{-4mm}
\end{figure}

\textbf{Our Contributions.}
%A large body of previous work focuses on finding top-$k$ nodes with highest betweenness \cite{yoshida2014,riondato2014, mahmoodye2016}. Others aim to optimize centrality by modifying the network efficiently \cite{parotsidis2016centrality, crescenzi2015}. 
% Our work differs from existing studies in two important ways. %First, our objective is to maximize coverage centrality, where as, the objectives in the existing work is minimizing closeness centrality of one node, probabilistic shortest path distances from one node or maximizing bewteenness centrality of one node. Second, previous work is focused in optimizing centrality of one particular node, whereas our problem considers coverage centrality\todo{more?} of a set of nodes. Note that, the group betweenness centrality is not just the sum of the betweenness centralities of individual nodes. 
%First, we investigate centrality from a control/design perspective, improving the coverage centrality of vertices. Second, we focus on a more general version of the problem, where a target set of nodes is given as input.  %Note that, the group coverage centrality is not the sum of the coverage centralities of vertices (see Equation \ref{eqn::coverage_centrality}).
%In this paper, we propose maximizing coverage centrality of a given set of nodes. To the best of our knowledge, this is the first work which addresses this problem. Furthermore\todo{check} we show how our approaches can be extended to optimize other group centralities.
The main contributions of this paper can be summarized as follows:
\begin{itemize}
    \item We study several variations of a novel network design problem, the group centrality optimization, and prove that they are NP-hard as well as APX-hard. %even to approximate by a constant. %We also show a strong inapproximibility result, that it is NP-hard to have a constant-factor approximation for the problem.
    \item We propose a simple greedy algorithm and even faster sampling algorithms for group centrality maximization.
    
    \item We show the effectiveness of our algorithms on several real datasets and also prove that the proposed solutions are optimal for a constrained version of the problem.
    %\item We also study how our methods generalizes to directed graphs. %We show that our approaches are applicable to group betweenness centrality maximization.    
    %further study a constrained yet realistic version of the problem that is APX-hard. Our proposed greedy and the randomized approach have a nearly optimal constant factor and probabilistic approximation guarantees respectively. 
\end{itemize}

\iffalse
While centrality of a set of nodes can be defined in many different ways, we consider the coverage centrality and beyond \todo{check}. More specifically, we optimize the initial coverage centrality of a set of nodes by adding edges from them to others.  We define the coverage centrality as follows: 
 If $P_{st}$ denotes the set of vertices which are on the shortest paths (multiple shortest paths might exist) between $s$ and $t$ and $s,t \notin P_{st}$. The coverage centrality  of a vertex (in terms of shortest paths) is defined as 
\begin{equation*}
    C(v) = |\{ (s,t) \in (V\times V)| v \in P_{st}\}|
\end{equation*}
$C(v)$ implies the number of pairs such that at least one shortest path between such pair goes through vertex $v$ and $V$ is set of vertices in the graph. 
The coverage centrality of a given set $X$ can be defined in similar way:
\begin{equation*}
    C(X) = |\{ (s,t) \in (V\times V)| v \in P_{st}, v\in X \}|
\end{equation*}

\fi

%1.5page-------------------------------------------------------------------------------
\section{Problem Definition}
\label{sec:problem_definition}

%In this section, we formalize the coverage centrality optimization (CCO) problem. %We assume that the network undirected and unweighted, but our methods can be generalized to directed networks (see Appendix). 
We assume $G(V,E)$ to be an undirected\footnote{We discuss how our methods can be generalized to directed networks in the Appendix.} graph with sets of vertices $V$ and edges $E$.  %We use network/graph and vertex/node interchangeably throughout this paper. 
A shortest path between vertices $s$ and $t$ is a path with minimum distance (in hops) among all paths between $s$ and $t$, with length $d(s,t)$. %Note that the distance between two vertices is in hop counts as the graph is unweighted. 
By convention, $d(s,s)=0$, for all $s\in V$. Let $P_{st}$
%\footnote{ Note that $P_{st}$ is not a shortest path from $s$ to $t$, but rather the set of all vertices that are on shortest paths from $s$ to $t$.} 
denote the set of vertices in the shortest paths (multiple ones might exist) between $s$ and $t$ where $s,t \notin P_{st}$. We define $Z$ as the set of candidate pairs of vertices, $Z\subseteq V\setminus X\times V\setminus X$, which we want to cover. The \textit{coverage centrality}  of a vertex is defined as:
\begin{equation}
   % C(v) = |\{ (s,t) \in Q| v \in P_{st}, s<t,s\neq v, t\neq v\}|
   C(v) = |\{ (s,t) \in Z| v \in P_{st},s\neq v, t\neq v\}|
   %C(v) = |\{ (s,t) \in Z| v \in P_{st}\}|
\end{equation}
$C(v)$ gives the number of pairs of vertices with at least one shortest path going through (i.e. covered by) vertex $v$. 
%($s<t$ implies ordered pairs). 
The \textit{coverage centrality} of a set of vertices $X \subseteq V$ is defined as:
\begin{equation}
    \begin{small}
    %C(X) = |\{ (s,t) \in V\times V| v \in P_{st}, v\in X \wedge s,t\notin X, s<t\}|
    C(X) = |\{ (s,t) \in Z| v \in P_{st}, v\in X \wedge s,t\notin X\}|
    %C(X) = |\{ (s,t) \in Z| v \in P_{st}, v\in X\}|
    \end{small}
    \label{eqn::coverage_centrality}
\end{equation}

%Additionally, we define \textbf{``cover"} for better explanation. If $v\in P_{st}$, then $v$ ``covers" the shortest path (or simply path) between $s$ and $t$. Similarly $X$ ``covers" the same path when $X \cap P_{st} \neq \Phi$.

%\textbf{Definition of Cover:} 
A set $X$ covers a pair ($s,t$) iff $X\cap P_{st} \neq \varnothing$, i.e., at least one vertex in $X$ is part of a shortest path from $s$ to $t$. 
\iffalse %----------------------constraint 1
We assume that edges are added from the target set $X$ to the remaining nodes, i.e. edges in given candidate set $\Gamma$ have the form $(a,b)$ where $a\in X$ and $b \in V\setminus X$. This is a reasonable assumption in many applications. For instance, in online advertising, adding links to a third-party page gives the third-party control over the traffic towards the advertiser's content, which is undesirable. Similarly, in monitoring scenarios, adding links between arbitrary nodes might disrupt the network (e.g. cause congestion) or possibly favor an adversary with similar monitoring purposes.
\fi
%\textbf{Maximizing $C(X)$ by adding edges:} 
Our goal is to maximize the coverage centrality of a given set $X$ over a set of pairs $Z$ by adding edges from a set of candidate edges $\Gamma$ to $G$. For instance, in our online advertising example, $X$ are pages owned by the car rental company, $Z$ are pairs of pages related to travel and cars, and $\Gamma$ are potential links that can be created to increase the traffic through the pages in $X$. Let $G_m$ denote the modified graph after adding edges $E_s\subseteq \Gamma$, $G_m=(V,E\cup E_s)$. We define the coverage centrality of $X$ (over pairs in $Z$) in the modified graph $G_m$ as $C_m(X)$.

\begin{table} [t]
\centering
\small
\vspace{-2mm}
\begin{tabular}{| c | c |}
\hline
\textbf{Symbols} & \textbf{Definitions and Descriptions}\\
\hline
$d(s,t)$ & Shortest path (s.p.) distance between $s$ and $t$\\
\hline
$n$ & Number of nodes in the graph\\
\hline
$m$ & Number of edges in the graph\\
\hline
$G(V,E)$ & Given graph (vertex set $V$ and edge set $E$)\\
\hline
$X$ & Target set of nodes\\
\hline
$C(v), C(X)$ & Coverage centrality of node $v$, node set $X$\\
\hline
$\Gamma$& Candidate set of edges\\
\hline
$k$& budget \\
\hline
$P_{st}$ & The set of nodes on the s.p.s between $s$ and $t$\\
\hline
$G_m, C_m$ & Modified graph and modified centrality\\
\hline
$Z$ & Pairs of vertices to be covered\\
\hline
$m_u$ & Number of uncovered pairs, $|Z|$\\
\hline
\end{tabular}
\caption { Frequently used symbols}\label{tab:table_symbol}
%\vspace{-3mm}
\end{table}

\begin{problem}\textbf{Coverage Centrality Optimization (CCO):}
Given a network $G=(V,E)$, a set of vertices $X\subset V$, a candidate set of edges $\Gamma$, a set of vertex pairs $Z$ and a budget $k$, find a set of edges $E_s\subseteq \Gamma$, such that $|E_s|\leq k$ and $C_m(X)$ is maximized.
\label{def:CCO}
\end{problem}

For simplicity, in the rest of the paper, we assume $Z=V\setminus X \times V\setminus X$ unless stated otherwise.  %All the results hold for any general given set of node pairs $Z$ with small variation. 
As a consequence: 
\begin{equation}
    \begin{small}
    %C(X) = |\{ (s,t) \in V\times V| v \in P_{st}, v\in X \wedge s,t\notin X, s<t\}|
    C(X) = |\{ (s,t) \in V\setminus X \times V\setminus X | v \in P_{st}, v\in X, s<t\}|
    %C(X) = |\{ (s,t) \in Q| v \in P_{st}, v\in X \wedge s,t\notin X\}|
    \end{small}
    \label{eqn::coverage_centrality2}
\end{equation}
where $s<t$ implies ordered pairs of vertices. 
% An example
Fig.~\ref{fig:example1} shows a solution for the CCO problem with budget $k=1$ for an example network where the target set $X=\{d,f\}$ and the candidate set $\Gamma= \{(d,a), (d,b), (f,b)\}$. %The addition of edge $(d,a)$, which is optimal, increases the centrality of $X$ from $C(X)=0$ to $C_m(X)=3$. More specifically, node $d$ belongs to $P_{ac}$, $P_{ae}$ and $P_{af}$ in $G_m$. Conversely, adding the edge $(d,b)$ would lead to $C_m(X) = 0$.

\begin{figure}[h]
    \centering
    %\vspace{-8mm}
    \subfloat[Initial graph]{
        \includegraphics[keepaspectratio, width=.15\textwidth]{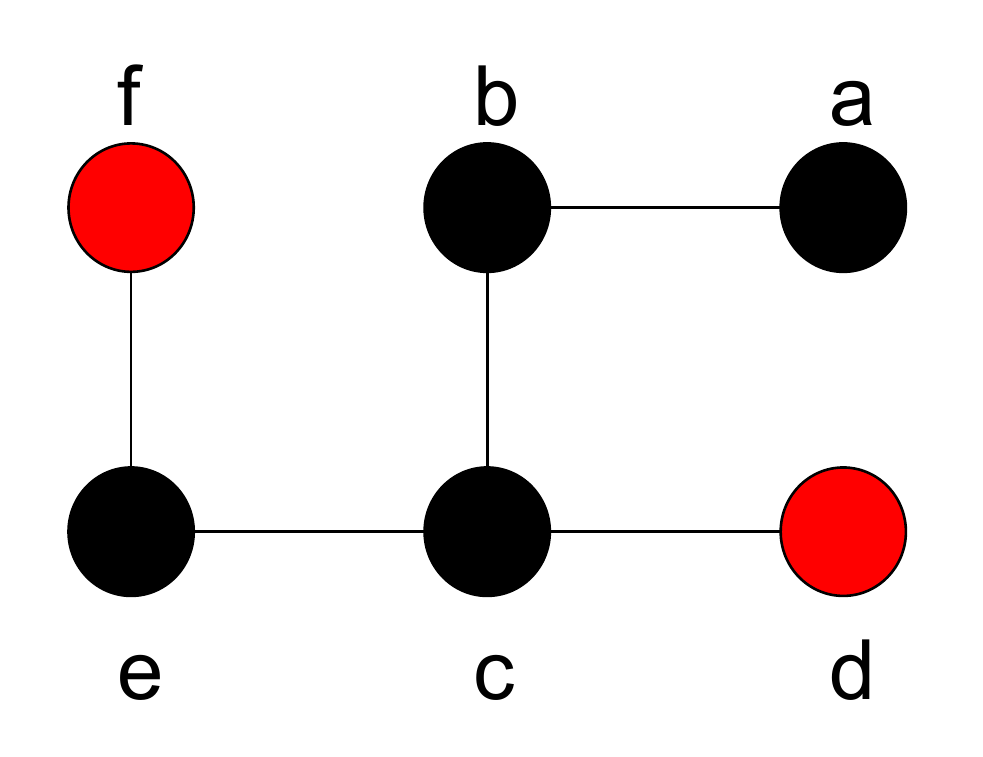}
       \label{fig:ex11}
    }
    \subfloat[Modified graph]{
       \includegraphics[keepaspectratio, width=.15\textwidth]{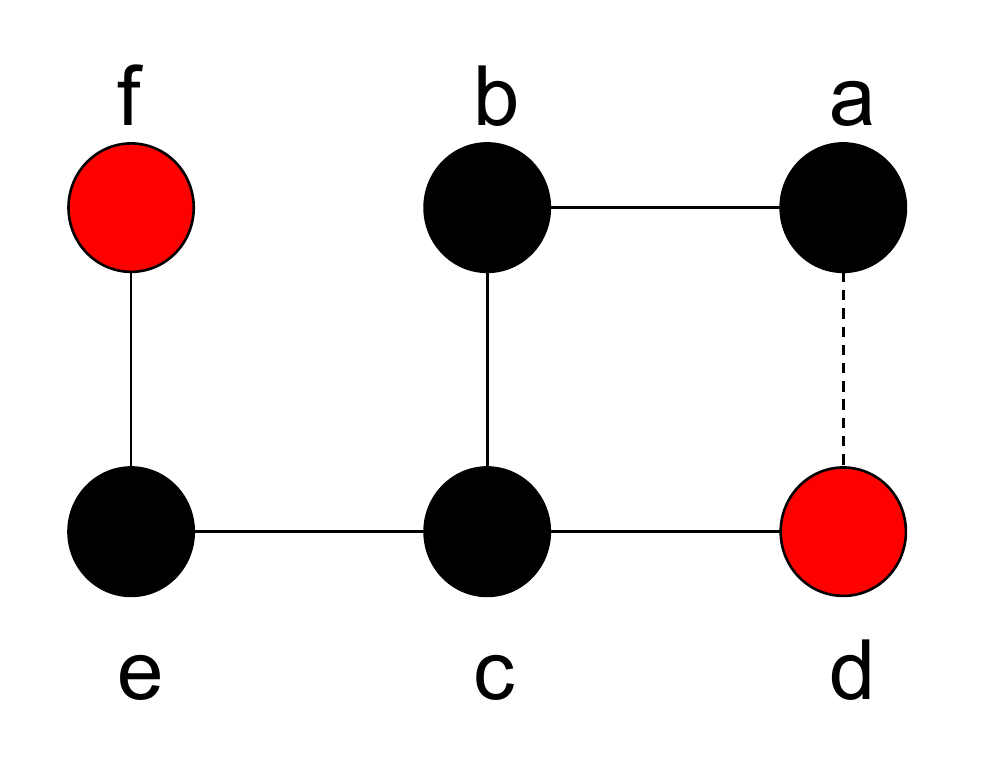}
       \label{fig:ex12}
    }
  %\vspace{-1mm}
     \caption{Example of Coverage Centrality Optimization problem. We want to optimize the centrality of $\{d,f\}$ with a budget of one edge from the candidates $\{(d,a), (d,b), (f,b)\}$. 
     %(a) represents the initial graph and (b) represents the modified graph with additional edge $(d,a)$. In (a) and (b), the coverage centrality of $\{d,f\}$ is $0$ and $6$ respectively. 
     The coverage centrality of $\{d,f\}$ is $0$ in the initial graph (a) and $3$ in the modified graph (b). Node $d$ belongs to the shortest paths between $(a,e)$, $(a,c)$ and $(a,f)$ in (b).\label{fig:example1}}
 %    \vspace{-2mm}
\end{figure}

%%%%%%%% DEFINITION OF group betweenness optimization

Similarly, we can also formulate the group betweenness centrality optimization problem. Given a vertex set $X \subseteq V$, its \textit{group betweenness centrality}  is defined as:
\begin{equation}
   % C(v) = |\{ (s,t) \in Q| v \in P_{st}, s<t,s\neq v, t\neq v\}|
   B(X) = \sum_{s,t\in V\setminus X} \frac{\sigma_{s,t}(X)}{\sigma_{s,t}}
   %C(v) = |\{ (s,t) \in Z| v \in P_{st}\}|
\end{equation}
where $\sigma_{s,t}$ is the number of shortest paths between $s$ and $t$, $\sigma_{s,t}(X)$ is the number of shortest paths between $s$ and $t$ passing through $X$. We define the group betweenness centrality of $X$ in the modified graph $G_m$ as $B_m(X)$. 

\begin{problem}\textbf{ Betweenness Centrality Optimization (BCO):}
Given a network $G=(V,E)$, a node set $X\subset V$, a candidate edge set $\Gamma$, a set of node pairs $Z$ and a budget $k$, find a set of edges $E_s\subseteq \Gamma$, such that $|E_s|\leq k$ and $B_m(X)$ is maximized.
\label{def:BCO}
\end{problem}

In this paper, we focus on the CCO problem. However, all the results described can be mapped to the BCO problem with small changes. To avoid redundancy and due to the space constraints, we omit similar results for BCO.

%\section{Theory}
\section{Hardness and Inapproximability}
\label{sec:hardness}

%describe np hardness
%In this section we present several interesting hardness, inapproximability and approximibility results for the general and constrained version of CCO. First we prove CCO is NP-hard by providing a reduction from \textit{Set Cover}. CCO is NP-hard even when the size of the target set is $1$ and $Z=V\times V$.

This section provides complexity analysis of the CCO problem. We show that CCO is NP-hard as well as APX-hard. More specifically, CCO cannot be approximated within a factor grater than $(1- \frac{1}{e})$. %This result offers a broader perspective on the hardness of network design problems.

%This result not only certifies that there are no efficient algorithms with constant quality guarantees for our general problem, but also offers a broader perspective on the hardness of network design problems.

%This section provides a theoretical analysis of the Coverage Centrality Optimization (CCO) problem. First, we show that COO is NP-hard, even to approximate by a constant. Next, we study a constrained version of the problem that allows a simple and efficient approximation algorithm with quality guarantees. We show that the proposed algorithm is optimal, in the sense that there is no polynomial algorithm that achieves a better approximation ratio. 

%The main challenge %of the problem comes from the
%lies in the interdependence among the $k$ edges to be added, which leads to a combinatorial number of possible solutions. %Note that in this problem we care only about how to maximize the \textbf{number} of shortest paths that are covered by the nodes in $X$. We do not aim to reduce the lengths of shortest paths.  

%%%%%%%%%%%%%%%%%%%%%%%%%%%%%%%%%%%%%%%         FIGURE
\begin{figure}[t]
    \centering
    \vspace{-2mm}
   \includegraphics[width=.35\textwidth]{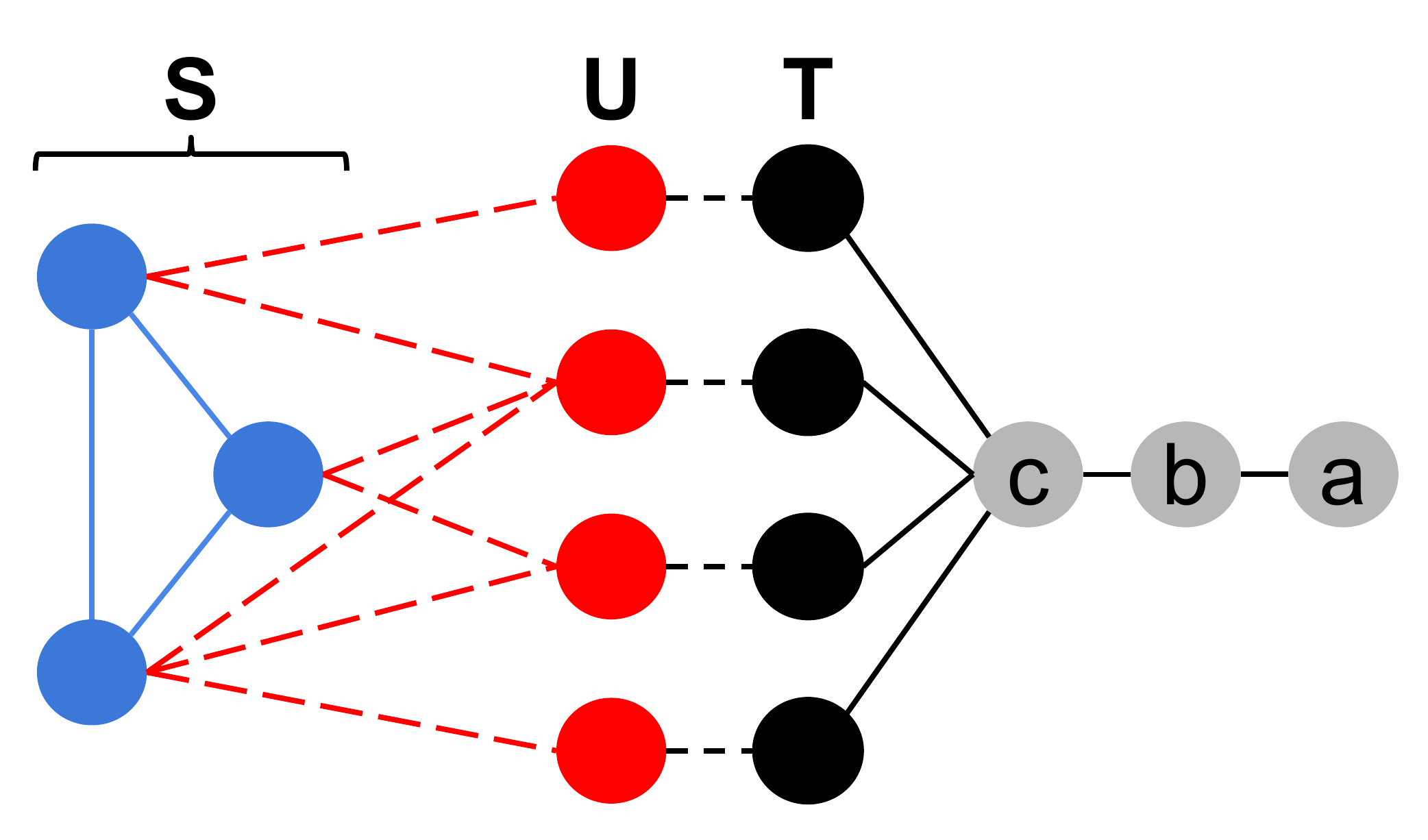}
   \vspace{-1mm}
    \caption{ Example of reduction from Set Cover to Coverage Centrality Optimization%: the red and blue nodes belong to set $U$ and $S$ respectively 
    , where $|U|=4$ and $|S|=3$. % The red edges are when the element belongs to a particular subset as in Set Cover problem. The black edges are further added as per construction. 
     Target set $X=\{a\}$ and candidate edges $\Gamma$ connect $a$ to nodes in $S$.\label{fig:nphard}}
%\vspace{-1mm}
 \end{figure}

%\begin{thm}\label{thm:nphard} CCO is NP-hard even when $|X|=1$ and $Z=V\times V$.
\begin{thm}\label{thm:nphard} The CCO problem is NP-hard.
\end{thm}

\begin{proof}
%See the Appendix.
%We outline a reduction from the Set Cover problem. 
Consider an instance of the NP-complete Set Cover problem, defined by a collection of subsets $S_{1},S_{2},...,S_{m}$ for a universal set of items $U=\{ u_{1},u_{2},...,u_{n} \}$. The problem is to decide whether there exist $k$ subsets whose union is $U$. To define a corresponding CCO instance, we construct an undirected graph with $m+2n+3$ nodes: there are nodes $i$ and $j$ corresponding to each set $S_{i}$ and each element $u_{j}$ respectively, and an undirected edge $(i,j)$ whenever $u_{j}\in S_{i}$. Every $S_{i}$ has an edge with $S_{j}$ when $i \neq j$ and  $i,j \in {1,2,...,m}$. The set $T =\{t_1,t_2,...,t_n\}$ is a copy of set $U$ where $u_{i}$ is connected to the corresponding $t_{i}$ for all $i \in {1,2,...,n}$. Three more nodes ($a,b$ and $c$) are added to the graph where $a$ is in $X$% (the given set of nodes in CCO)
. Node $c$ is connected to $t_{i}$ for all $i \in {1,2,...,n}$. Node $b$ is attached to $a$ and $c$. Figure \ref{fig:nphard} shows an example of this construction. The reduction clearly takes polynomial time. The candidate set $\Gamma$ consists of the edges between $a$ and set $S$. We prove that CCO of a given singleton set is NP-hard by maximizing the Coverage Centrality (CC) of the node $a$. Current CC of $a$ is $0$ by construction. 

%\textbf{Claim:} 
A set $S'\subset S$, with $|S'|\leq k$ is a set cover iff the CC of $a$ becomes $n+m+k$ after adding the edges between $a$ and every node in $S'$. 
%Table \ref{tab:table_nphard} shows the possibility of pairs being covered by node $a$ after the edges are added between node $a$ and any node in $S' \subset S$. The ``F"s and ``T"s show the impossibility and possibility respectively (explanation of ``F"s are omitted due to brevity).
%\textbf{Only if: } 
Assume that $S'$ is a set cover and edges are added between node $a$ and every node in $S'$. %$T_1:$ $d(b,y)$ for any $y\in S$  improves and becomes $\leq 3$. $T_2:$ $d(c,y)$ for any $y\in S'$ improves and becomes $\leq 3$. $T_3:$  $d(b,y)=3$ for any $y\in U$. As the edges are attached to set cover, all the nodes in $U$ will be in distance $3$. So, CC of $a$ increases by $m$, $k$ and $n$ in $T1$, $T2$, and $T3$ respectively and becomes $m+n+k$.
Then the CC of $a$ improves by $m+n+k$ as shortest paths between pairs $(b,s)$, $\forall s \in S$, $(b,u)$, $\forall u\in U$, and $(c,p)$, $\forall p\in S'$, will now pass through $a$. On the other hand, assume that the
%\textbf{If: } 
CC of node $a$ is $m+n+k$ after adding edges between  $a$ and any set $S'\subset S$. %$T_1$ and $T_2$ are same as in ``only if" and CC of $a$ becomes $m+k$. $T3:$ To have $n$ more shortest paths via node $a$, every node in $U$ should be reachable from $b$ via node $a$ with distance $3$. That is only possible when $S'$ is set cover. Hence, the claim is proved.
It is easy to see that $m+k$ extra pairs will have their shortest paths covered by $a$. However, the only way to add another $n$ pairs is by making $S'$ a set cover. 
%%%%%%%%%% Omitting the explanation for Fs
\iffalse
Case A) $F_{10}$ and $F_{12}$: the distances are $0$. Case B) $F_{1}$, $F_{9}$ and $F_{11}$: the distances between the nodes are $1$. Case C) $F_{6}$, $F_{7}$ and $F_{8}$: the distances between the nodes are $2$ (if it is same node then distance is $0$). Case D) $F_2$: The distance from any node in $S$ to any node in $U$ are either $1$ or $2$ as nodes in $S$ are connected. Case E) $F_3$: The distance from any node in $S$ to any node in $T$ are either $2$ (if attached to same node in $U$) or at most $3$. Case F) $F_4$: The distance from any node in $U$ to any node in $U$ are either $0$ (same node), $2$ (attached to same node in $S$) or at most $3$. Case G) $F_5$: The distance from any node in $U$ to any node in $T$ are either $1$ (attached) or at most $3$.
For obvious reasons, CC of node $a$ will not increase for case A and B. It is easy to verify that adding edges between the node $a$ and nodes in $S'$ do not create any shortest path via $a$ for case C, D, E, F and G. 
\fi
\end{proof}

Given that computing an optimal solution for CCO is infeasible in practice, a natural question is whether it has a polynomial-time approximation. The next theorem shows that CCO is also NP-hard to approximate within a factor greater than $(1- \frac{1}{e})$. Interestingly, different from its search counterpart \cite{yoshida2014}, CCO is not submodular (see Section \ref{sec:general_settings}). These two results provide strong evidence that, for group centrality, network design is strictly harder than search.

%Given that computing an optimal solution for CCO is infeasible in practice, a natural question is whether there is a polynomial-time approximation for the problem. The next theorem shows that CCO is also NP-hard to approximate by any constant. Interestingly, such a proof requires a more elaborate construction than the one used in the last theorem. The intuition behind this fact is that graph modification problems have a stronger combinatorial nature than traditional coverage problems. In particular, while maximum coverage can be approximated by a simple greedy algorithm~\cite{nemhauser1978}, the same does not hold for our problem.  

%The next theorem shows that CCO is also NP-hard to approximate by any constant. Interestingly, such a proof requires a more elaborate construction than the one used in the last theorem. The intuition behind this fact is that graph modification problems have a stronger combinatorial nature than traditional coverage problems. In particular, while maximum coverage can be approximated by a simple greedy algorithm~\cite{nemhauser1978}, the same does not hold for our problem. 
 
 %\begin{thm}\label{thm:inapprox} CCO is NP-hard to approximate by a constant.
%\end{thm}

\begin{thm}\label{thm:inapprox} CCO cannot be approximated within a factor greater than $(1- \frac{1}{e})$.
\end{thm}

\begin{proof}
We give an $L$-reduction \cite{williamson2011design} from the maximum coverage (MSC) problem with parameters $x$ and $y$. Our reduction is such that following two equations are satisfied:
 \begin{equation}
     OPT(I_{CCO}) \leq xOPT(I_{MSC})
     \vspace{-2mm}
 \end{equation}
  \begin{equation}
     OPT(I_{MSC})-s(T^M) \leq y(OPT(I_{CCO})-s(T^C))
 \end{equation}
where $I_{MSC}$ and $I_{CCO}$ are problem instances, and $OPT(Y)$ is the optimal value for instance $Y$. $s(T^M)$ and $s(T^C)$ denote any solution of the MSC and CCO instances respectively. If the conditions hold and CCO has an $\alpha$ approximation, then MSC has an $(1-xy(1-\alpha))$ approximation. However, MSC is NP-hard to approximate within a factor greater than $(1-\frac{1}{e})$. It follows that $(1-xy(1-\alpha))< (1-\frac{1}{e})$, or, $\alpha < (1-\frac{1}{xye})$ \cite{crescenzi2015}. So, if the  conditions are satisfied, CCO is NP-hard to approximate within a factor greater than $(1-\frac{1}{xye})$. 

 We use the same construction as in Theorem \ref{thm:nphard}. For CCO, the set $Z$ contains pairs in the form $(b,u)$, $u \in U$. Let the solution of $I_{CCO}$ be $s(T^C)$. The centrality of node $a$ will increase by $s(T^C)$ to cover the pairs in $Z$. Note that $s(T^C)= 2s(T^M)$ from the construction (as the graph is undirected, the covered pair is unordered). It follows that both the conditions are satisfied when $x=2$ and $y=\frac{1}{2}$. So, CCO is NP-hard to approximate within a factor grater than $(1-\frac{1}{e})$.
 
 \end{proof}

Theorem \ref{thm:inapprox} shows that there is no polynomial-time approximation better than $(1- \frac{1}{e})$ for CCO. Given such an inapproximation result, we propose an efficient greedy heuristic for our problem, as discussed in the next section.

\section{Algorithms}
\label{sec:greedy}
%
%In this section, we first present an efficient greedy algorithm. % with a constant-factor approximation for the constrained version of  Coverage Centrality Optimization (CCO). 
%We further improve the performance of the greedy algorithm in order to scale it to large networks via a novel sampling scheme.

\subsection{Greedy Algorithm (GES)}
\label{sec:greedy_algo}

 Algorithm \ref{algo:GES} (GES) is a simple greedy strategy that selects the best edge to be added in each of the $k$ iterations,  where $k$ is the budget. Its most important steps are $2$ and $7$. In step $2$, it computes all-pair-shortest-paths in time $O(n(m+n))$. Next, it chooses, among the candidate edges $\Gamma$, the one that maximizes the marginal coverage centrality gain of $X$  (step $7$), which takes $O(|\Gamma|n^2)$ time. After adding the best edge, the shortest path distances are updated. Then, the algorithm checks the pairwise distances in $O(n^2)$ time (step $9$). The total running time of GES is $O(n(m+n)+ k|\Gamma|n^2)$.

%The Algorithm
\begin{algorithm}[ht]
\small
 \caption {Greedy Edge Set (GES)\label{algo:GES}}
\begin{algorithmic}[1] 
 \REQUIRE Network $G=(V,E)$, target node set $X$, Candidate set \newline of edges $\Gamma$, Budget $k$
 \ENSURE A subset $E_s$ from $\Gamma$ of $k$ edges 
\STATE $E_s\leftarrow\emptyset$
\STATE Compute all pair shortest paths and store the distances
\WHILE {  $|E_s|\leq k$ }
\FOR{$e \in \Gamma \setminus E_s$}
\STATE $\textit{Count(e)} \leftarrow$ $\#$ newly covered pairs after adding $e$
\ENDFOR
\STATE $e^*\leftarrow \argmax_{e\in \Gamma \setminus E_s}\{Count(e)\}$
\STATE $E_s\leftarrow E_s\cup e^*$ and $E\leftarrow E\cup e^*$ 
\STATE Update the shortest path distances
\ENDWHILE
 \RETURN $E_s$
\end{algorithmic}
\end{algorithm}

\iffalse
\begin{algorithm}[ht]
\scriptsize
 \caption {BestEdge\label{algo:GES_sub}}
\begin{algorithmic}[1] 
 \REQUIRE Network $G=(V,E)$, target node set $X$, Candidate set of edges $\Gamma$
 \ENSURE An edge from $\Gamma$
%\WHILE {  $|E_s|\leq k$ }
\FOR{$e \in \Gamma$}
\STATE $\textit{Count(e)} \leftarrow$ $\#$ newly covered pairs after adding $e$
\ENDFOR
\STATE $e^*\leftarrow \argmax_{e\in \Gamma}\{Count(e)\}$
\STATE Update the shortest path distances
\RETURN $e^*$
%\ENDWHILE
\end{algorithmic}
\end{algorithm}

\fi

\iffalse
\begin{figure}[t]
    \centering
    \vspace{-5mm}
    \subfloat[Initial graph]{
        \includegraphics[keepaspectratio, width=.2\textwidth]{Figures/Example/example2.pdf}
       \label{fig:ex21}
    }
    \subfloat[After first step]{
       \includegraphics[keepaspectratio, width=.2\textwidth]{Figures/Example/example3.pdf}
       \label{fig:ex22}
    }
       
     \caption{ \small Illustrative example of greedy (GES) to solve CCO. Here we want optimize the centrality of $\{d,f\}$, the budget is chosen from the candidate set $ \{(d,a), (d,b), (f,b)\}$. (a) represents the initial graph and (b) represents the modified graph after first step of GES. In (a) and (b), the coverage centrality of $d$ is $0$ and $6$ respectively. Node $d$ belongs to the shortest paths between pair $(a,e)$, pair $(a,c)$ and pair $(a,f)$ in (b).}\label{fig:example2}
     \vspace{-4mm}
\end{figure}
\fi

%\textbf{Example:} 
We illustrate the execution of GES on the graph from Figure \ref{fig:ex11} for a budget $k=2$, a candidate set of edges $\Gamma= \{(d,a),$ $(d,b), (f,b)\}$, and a target set $X=\{d,f\}$. Initially, adding $(d,a), (d,b)$ and $(f,b)$ increases the centrality of $X$ by $3$, $0$, and $2$, respectively, and thus $(d,a)$ is chosen. In the second iteration, $(d,b)$ and $(f,b)$ increase the centrality of $X$ by $0$ and $1$, respectively, and $(f,b)$ is chosen.

\subsection{Sampling Algorithm (BUS)}
\label{sec::approximation_algorithms}

The execution time of GES increases with $|\Gamma|$ and $m$. In particular, if $m = O(n^2)$ and $|\Gamma|=O(n)$, the complexity reaches $O(n^3)$, which is prohibitive for large graphs. To address this challenge, we propose a sampling algorithm % In this section, we show how sampling can be applied in order scale our algorithm to real-world graphs, with sizes in the order of millions of vertices. In particular, our sampling algorithm 
that is nearly optimal, regarding each greedy edge choice, with probabilistic guarantees (see Section \ref{sec:analysis_sampling}).
%GES computes all-pair-shortest-paths, %computation and such computational and storage requirements 
%which renders it infeasible within a reasonable time even for moderate-size networks ---with tens of thousands of nodes. Our goal, however, is to design an algorithm that will scale to real-world networks, with sizes in the order of millions of nodes.
%Here, we present a uniform random sampling based method that reduces the high complexity of GES and produces nearly optimal result with probabilistic guarantees.  
 Instead of selecting edges based on all the uncovered pairs of vertices, our scheme does it based on a small number of sampled uncovered pairs. %GES computes the best edge at each step as the one that covers the maximum number of uncovered pairs.
This strategy allows the selection of edges with probabilistic guarantees using a small number of samples, thus ensuring scalability to large graphs. We show that the error in estimating the coverage based on the samples is small. %with high probability.

%Before presenting the algorithm formally, we demonstrate the analysis of approximation via this sampling method. We achieve a probabilistic approximation of the optimal solution. Our approach is very similar to the approaches described in \cite{mahmoodye2016,yoshida2014} \todo{fix}.

%\subsubsection*{Analysis}

%%%%%%%%%%%%%%%%%%%%%%%%%%%%%%%%%%%%%%%

Algorithm \ref{algo:BUS} (Best Edge via Uniform Sampling, or BUS) is a sampling strategy to select the best edge to be added in each of the $k$ iterations based on the sampled uncovered node pairs. For each pair of samples, we compute the distances from each node in the pair to all others. These distances %are used to
estimate the number of covered pairs after the addition of one edge. In Section \ref{sec:analysis_sampling}, we provide a theoretical analysis of the approximation achieved by BUS.

The costliest steps of our algorithm are 4-7 and 8-10. Steps 4-7, where the algorithm performs shortest-path computations, take $O(q(n+m))$ time. Next, the algorithm estimates the additional number of shortest pairs covered by $X$ after adding each of the edges based on the samples (steps 8-10) in $O(|\Gamma|q^2)$ time. Given such an estimate, the algorithm chooses the best edge to be added (step 11). The total running time of BUS is $O(kq(m+n)+ k|\Gamma|q^2)$.

%%%%%%%%%%%%%%%%%%%%%%%%%%% ALGO

\begin{algorithm}[t]
\small
 \caption {Best Edge via Uniform Sampling (BUS)\label{algo:BUS}}
\begin{algorithmic}[1] 
\REQUIRE Network $G=(V,E)$, target node set $X$, Candidate set \newline of edges $\Gamma$, Budget $k$
 \ENSURE A subset $\gamma$ from $\Gamma$ of $k$ edges  
\STATE Choose $q$ pairs of vertices in $Q$ from $M_u$
\STATE $\gamma \leftarrow \emptyset$
\WHILE {  $|\gamma|\leq k$ }
\FOR {  $(s,t)\in Q$ }
\STATE Compute and store s.p. distance $d(s,v)$ (for all $v\in V$)
\STATE  Compute and store s.p. distance $d(t,v)$ (for all $v\in V$)
\ENDFOR
\FOR{ $e \in \Gamma \setminus \gamma$}
\STATE $score_{e}\leftarrow \#$ newly covered pairs after adding $e$
\ENDFOR
\STATE $e^* \leftarrow \argmax_{e'\in \Gamma \setminus \gamma}\{score_{e'}\}$
%\STATE Adjust $d(s,v)$ and $d(t,v)$, where $(s,t)\in Q$ and $v\in V$
\STATE $ \gamma \leftarrow \gamma\cup e^*$ and $E\leftarrow E\cup e^*$ 
\ENDWHILE
\STATE Return $\gamma$
%\ENDIF 
\end{algorithmic}
\end{algorithm}

\section{Analysis}

%In the previous section, we described an efficient algorithm for CCO. Based on the inaproximality result from Theorem \ref{thm:inapprox}, such an algorithm cannot provide a constant-factor approximation for our general problem. Nevertheless, under some realistic assumptions, we show that the described algorithm provides a constant-factor approximation for a modified version of CCO. More specifically, our approximation guarantees are based on the addition of two extra constraints to the general problem described in Section \ref{sec:problem_definition}. 

In the previous section, we described a greedy heuristic and an efficient algorithm to approximate the greedy approach. Next, we show, under some realistic assumptions, the described greedy algorithm provides a constant-factor approximation for a modified version of CCO. More specifically, our approximation guarantees are based on the addition of two extra constraints to the general CCO described in Section \ref{sec:problem_definition}. 

\subsection{Constrained Problem}
\label{sec:analysis_const}

The extra constraints, $S^1$ and $S^2$, considered are the following: (1) $S^1$: We assume that edges are added from the target set $X$ to the remaining nodes, i.e. edges in  a given candidate set $\Gamma$ have the form $(a,b)$ where $a\in X$ and $b \in V\setminus X$ \cite{crescenzi2015}; and (2) $S^2$: Each pair $(s,t)$ can be covered by at most one single newly added edge \cite{chaoji2012recommendations, meyerson2009}.

$S^1$ is a reasonable assumption in many applications. For instance, in online advertising, adding links to a third-party page gives away control over the navigation, which is undesirable. %Similarly, in monitoring scenarios, adding links between arbitrary nodes might disrupt the network (e.g. cause congestion) or possibly favor an adversary with similar monitoring purposes. 
$S^2$ is motivated by the fact that, in real-life graphs, vertex centrality follows a skewed distribution (e.g. power-law), and thus most of the new pairs will have shortest paths through a single edge in $\Gamma$. %For instance, in a social network, adding more than one friendship to a target user $v$ will not significantly increase its centrality, since the resulting paths through $v$ will still be long compared to the ones involving other high centrality users. 
In our experiments (see Table \ref{table:motivation} in Section \ref{sec:exp_rcco_vs_cco}), we show that, in practice, solutions for the constrained and general problem are not far from each other. %We experiment with two small networks: co-authorship (Netscience) and synthetic (Barabasi) data. 
Both constraints have been considered by previous work \cite{crescenzi2015,chaoji2012recommendations,meyerson2009}.
 Next, we show that COO under constraints $S^1$ and $S^2$, or RCCO (Restricted CCO), for short, is still NP-hard.

%\begin{cor}\label{cor:nphard_c} CCO under $S^1$ and $S^2$ is NP-hard.
\begin{cor}\label{cor:nphard_c} RCCO is NP-hard.
\end{cor} 
\begin{proof}
Follows directly from Theorem \ref{thm:nphard}, as the construction applied in the proof respects both the constraints.
 \end{proof}

%Notice that our inaproximability result (Theorem \ref{thm:inapprox}) applies a construction that violates the constraints $S^1$ and $S^B$. This motivates us to explore the existence of approximate algorithms for RCCO. 
 
 \subsection{Analysis: Greedy Algorithm}
 \label{sec:analysis_greedy}
 
 The next theorem shows that RCCO's optimization function is  monotone and submodular.  As a consequence, the greedy algorithm described in Section~\ref{sec:greedy_algo} leads to a well-known constant factor approximation of $(1-1/e)$~\cite{nemhauser1978}. 
 
 %Therefore, if $E_s$ is the edge set produced by GES, $f(E_s) \geq (1-1/e)OPT$, where $OPT$ denotes the optimal coverage.

%\begin{thm}\label{thm:submodular} The objective function $f(E_s)=C_m(X)$ in CCO is monotone and submodular under $S^1$ and $S^B$. 
\begin{thm}\label{thm:submodular} The objective function $f(E_s)=C_m(X)$ in RCCO is monotone and submodular. 
\end{thm}

\begin{proof}
\underline{Monotonicity:} Follows from the definition of a shortest path. Adding an edge $(u,v) \in E_s$ cannot increase $d(s,t)$ for any $(s,t)$ already covered by $X$. Since $u\in X$ for any $(u,v) \in E_s$, the coverage $C_m(X)$ is also non-decreasing. 

\underline{Submodularity:} We consider addition of two sets of edges, $E_a$ and $E_b$ where $E_a \subset E_b$, and show that $f(E_a \cup \{e\}) - f(E_a) \geq f(E_b\cup \{e\})-f(E_b)$ for any edge $e \in \Gamma$ such that $e \notin E_a$ and $e \notin E_b$. Let $F(A)$ be the set of node pairs $(s,t)$ which are covered by an edge $e \in A$ ($|F(E_s)| = C_m(X)$). Then $f(.)$ is submodular if $F(E_b\cup \{e\})\setminus F(E_b) \subseteq F(E_a\cup \{e\}) \setminus F(E_a)$. To prove this claim, we make use of $S^B$. Therefore, each pair $(s,t) \in F(E_b)$ is covered by only one edge in $E_b$. As $E_a \subset E_b$, adding $e$ to $E_a$ will cover some of the pairs which are already covered by $E_b\setminus E_a$. Then, for any newly covered pair $(s,t) \in F(E_b\cup \{e\})\setminus F(E_b)$, it must hold that $(s,t) \in F(E_a\cup \{e\})\setminus F(E_a)$.
%Addition of these edge sets, covers $A$ and $B$ among the initial uncovered pairs. As the pairs can be covered via one edge,  $A \subseteq B$ and $|A|\leq |B|$. Now consider the addition of edge $e$ where $e\notin A$ and $e \notin B$. As a result, $X$ covers pairs that either (1) belong to both $B$ and $A$, (2) belong to neither $B$ nor $A$, or (3) belong to  $B-A$. In the first two cases, for $E_a$ and $E_b$, there is either no gain (case $1$) or the gain is the same (case $2$) in terms of covering more pairs. Case $3$ leads to more gain by $e$ to $E_a$. So, $f(E_a \cup e)- f(E_a) \geq f(E_b \cup e)- f(E_b)$, and, $f(.)$ is submodular.
\end{proof}

Based on Theorem \ref{thm:submodular}, if $OPT$ is the optimal solution for an instance of the RCCO problem, GES will return a set of edges $E_s$ such that $f(E_s) \geq (1-1/e)OPT$. The existence of such an approximation algorithm shows that the constraints $S^1$ and $S^2$ make the CCO problem easier, compared to its general version. On the other hand, whether GES is a good algorithm for the modified CCO (RCCO) remains an open question. In order to show that our algorithm is optimal, in the sense that the best algorithm for this problem cannot achieve a better approximation from those of GES, we also prove an inapproximability result for the constrained problem. 

%Based on Theorem \ref{thm:submodular}, if $OPT$ is the optimal solution for an instance of the RCCO problem, GES will return a set of edges $E_s$ such that $f(E_s) \geq (1-1/e)OPT$. The existence of such an approximation algorithm shows that the constraints $S^1$ and $S^2$ make the CCO problem easier, compared to its general version. On the other hand, whether GES is a good algorithm for the modified CCO (RCCO) remains an open question. In order to show that our algorithm is almost optimal, in the sense that the best algorithm for this problem cannot achieve results far from those of GES, we also prove an inapproximability result for the constrained problem. 

%\begin{thm}\label{thm:approx_lower} CCO under $S^1$ and $S^2$ cannot be approximated within a factor greater than $(1- \frac{1}{4e})$.
\begin{cor}\label{thm:approx_lower} RCCO cannot be approximated within a factor greater than $(1- \frac{1}{e})$.
\end{cor}
 \begin{proof}
 Follows directly from Theorem \ref{thm:inapprox}, as the construction applied in the proof respects both the constraints.
 \end{proof}

 \iffalse
 \begin{proof}
  We give an $L$-reduction \cite{williamson2011design} from the maximum coverage (MSC) problem with parameters $x$ and $y$. Our reduction is such that following two equations are satisfied:
 \begin{equation}
     OPT(I_{RCCO}) \leq xOPT(I_{MSC})
 \end{equation}
  \begin{equation}
     OPT(I_{MSC})-s(T^M) \leq y(OPT(I_{RCCO})-s(T^C))
 \end{equation}
where $I_{MSC}$ and $I_{RCCO}$ are problem instances, and $OPT(Y)$ is the optimal value for instance $Y$. $s(T^M)$ and $s(T^C)$ denote any solution of the MSC and RCCO instances respectively. If the conditions hold and RCCO has an $\alpha$ approximation, then MSC has an $(1-xy(1-\alpha))$ approximation. However, MSC is NP-hard to approximate within a factor greater than $(1-\frac{1}{e})$. It follows that $(1-xy(1-\alpha))< (1-\frac{1}{e})$, or, $\alpha < (1-\frac{1}{xye})$ \cite{crescenzi2015}. So, if the  conditions are satisfied, RCCO is NP-hard to approximate within a factor greater than $(1-\frac{1}{xye})$. 

 We use the same construction as in Theorem \ref{thm:nphard}. For RCCO, the set $Z$ contains pairs in the form $(b,u)$, $u \in U$.
 
 Let the solution of $I_{RCCO}$ be $s(T^C)$. It is easy to see that the centrality of node $a$ will increase by $s(T^C)$ to cover the pairs in $Z$. Note that $s(T^C)= 2s(T^M)$ from the construction (as the graph is undirected, the covered pair is unordered). So, it follows that both the conditions are satisfied when $x=y=2$. So, RCCO is NP-hard to approximate within a factor grater than $(1-\frac{1}{4e})$.
\end{proof}

\fi

 Corollary \ref{thm:approx_lower} certifies that GES achieves the best approximation possible for the constrained CCO (RCCO) problem.

\subsection{Analysis: Sampling Algorithm}
\label{sec:analysis_sampling}

In Section \ref{sec::approximation_algorithms}, we presented BUS, a fast sampling algorithm for the general CCO problem. %BUS applies uniform sampling to select the best edge to be added at each iteration. 
Here, we study the quality of the approximation provided by BUS as a function of the number of sampled node pairs. The analysis will assume the constrained version of CCO (RCCO), but approximation guarantees regarding the general case will also be discussed.

Let us assume that $X$ covers a set $M_c$ of pairs of vertices. The set of remaining vertex pairs is $M_u$,  $M_u=\{(s,t)|s \in V,t\in V, s\neq t, X\cap P_{st}=\emptyset\},m_u=|M_u|=n(n-1)-|M_c|$.  We sample, uniformly with replacement, a set of ordered vertex pairs $Q$ ($|Q|=q$) from all vertex pairs ($M_u$) not covered by $X$. Let $g^q(.)$ denote the number of \emph{newly} covered pairs by the candidate edges based on the samples $Q$.
Moreover, for an edge set $\gamma \subset \Gamma$, let $X_i$ be a random variable which denotes whether the $i$th sampled pair is covered by any edge in $\gamma$. In other words, $X_i=1$ if the pair is covered and $0$, otherwise. Each pair is chosen with probability $\frac{1}{m_u}$ uniformly at random.

%the General Relativity and Quantum Cosmology collaboration network.% (with $5K$ nodes) and a synthetic (Barabasi) network (with similar size). 
%With only a few edges ($40$ and $70$, respectively), we can increase the coverage centrality of a target set containing $5$ nodes to $50\%$ and $33\%$ of their upper bound (UB). Our experiments are presented in Sec. \ref{sec:exp}.

%and the choice of one sample does not affect that of any other sample and thus, the samples are independent. 

\begin{lemma} \label{lemma:expectation}
Given a size $q$ sample of node pairs from $M_u$: %the expected average number of covered paths by $\gamma$ is an unbiased estimate of the average of covered paths in $M_u$:
\begin{equation*}
E(g^q(\gamma))= \frac{q}{m_u}f(\gamma)
\end{equation*}
%\vspace{1mm}
\end{lemma}
From the samples, we get $g^q(\gamma)=\Sigma_{i=1}^{q}X_i$.
By the linearity and additive rule, $E(g^q(\gamma))=\Sigma_{i=1}^{q}E(X_i)=q.E(X_i).$
As the probability $P(X_i)= \frac{f(\gamma)}{m_u} $ and $X_i$s are i.i.d., $E(g^q(\gamma))=\frac{q}{m_u}f(\gamma).$
%$g^q$ computes the number of \emph{newly} covered pairs by the candidate edges based on the sampled pairs. So, 
Also, let $f^q= \frac{m_u}{q}g^q$ be the estimated coverage. 

\begin{lemma} \label{lemma:anyE}
Given $\epsilon$ $(0<\epsilon<1)$, a positive integer $l$, a budget $k$, and a sample of independent uncovered node pairs $Q, |Q|=q$, where 
$q(\epsilon)\geq \frac{3m_u(l+k)log(|\Gamma|)}{\epsilon^2 \cdot OPT}$; then:
\begin{equation*}
Pr (|f^q(\gamma)- f(\gamma)| < \epsilon \cdot OPT) \geq 1- 2|\Gamma|^{-l}
\end{equation*}
For all $\gamma \subset \Gamma$, $|\gamma| \leq k$, where $OPT$ denotes the optimal coverage $(OPT= Max\{f(\gamma)|\gamma\subset \Gamma, |\gamma|\leq k\})$.
\end{lemma}
\begin{proof}
%To prove this, we begin with the following: \newline
Using Lemma \ref{lemma:expectation}:
\begin{equation*}
\begin{split}
Pr (|f^q(\gamma)- f(\gamma)| &\geq \delta \cdot f(\gamma))\\
Pr \Big(|\frac{q}{m_u}f^q(\gamma)- \frac{q}{m_u}f(\gamma)| &\geq  \frac{q}{m_u} \cdot \delta \cdot f(\gamma)\Big)\\ 
 Pr \Big(|g^q(\gamma)- \frac{q}{m_u}f(\gamma)| &\geq \frac{q}{m_u} \cdot \delta f(\gamma)\Big)\\ 
 Pr (|g^q(\gamma)- E(g^q(\gamma))| &\geq  \delta E(g^q(\gamma))) 
\end{split}
\end{equation*}
As the samples are independent, applying Chernoff bound:
\begin{equation*}
Pr \Big(|g^q(\gamma)- \frac{q}{m_u}f(\gamma)| \geq \frac{q}{m_u}\delta f(\gamma)\Big) \leq 2\exp\Big(-\frac{\delta^2}{3} \frac{q}{m_u}f(\gamma)\Big)
\end{equation*}
Substituting $\delta=\frac{\epsilon OPT}{f(\gamma)}$ and $q$:
\begin{equation*}
Pr (|f^q(\gamma)- f(\gamma)| \geq \epsilon \cdot OPT) \leq 2\exp \Big(-\frac{OPT}{f(\gamma)} (l+k)log(\Gamma)\Big)
\end{equation*}
Using the fact that $OPT\geq f(\gamma)$:
\begin{equation*}
Pr (|f^q(\gamma)- f(\gamma)| \geq \epsilon \cdot OPT) \leq 2|\Gamma|^{-(l+k)}
\end{equation*}

Applying the union bound over all possible size-$k$ subsets of $\gamma \subset \Gamma$ (there are $|\Gamma|^k$) we conclude the following:
\begin{equation*}
Pr (|f^q(\gamma)- f(\gamma)| \geq \epsilon \cdot OPT) < 2|\Gamma|^{-l}, \forall \gamma \subset \Gamma
\end{equation*}
\begin{equation*}
Pr (|f^q(\gamma)- f(\gamma)| < \epsilon \cdot OPT) \geq 1- 2|\Gamma|^{-l}, \forall \gamma \subset \Gamma \qedhere
\end{equation*} 
\end{proof}

Now, we prove our main theorem which shows an approximation bound of $(1-\frac{1}{e}- \epsilon)$ by Algorithm \ref{algo:BUS} whenever the number of samples is at least $q(\epsilon/2)= \frac{12m_u(l+k)log(|\Gamma|)}{\epsilon^2 \cdot OPT}$ ($l$ and $\epsilon$ are as in Lemma \ref{lemma:anyE}).

%%%%%%%%%%%%%%% Main Theorem %%%%%%%%%%%%%%%%%%%%%%%
\begin{thm}\label{thm:BUS_approx}
  Algorithm \ref{algo:BUS} ensures $f(\gamma)\geq (1-\frac{1}{e}- \epsilon)OPT $ with high probability $(1- \frac{2}{|\Gamma|^l})$ using at least $q(\epsilon/2)$ samples.% where $\gamma$ is its output.
\end{thm}
\begin{proof}
$f(.)$ is monotonic and submodular (Thm. \ref{thm:submodular}) and one can prove the same for $f^q(.)$. 
Given the following:
\begin{enumerate}
    \item Lemma \ref{lemma:anyE}: The number of samples is at least $q(\epsilon/2)$. So, with probability $1- \frac{2}{|\Gamma|^l}$, $f(\gamma) \geq  f^q(\gamma)- \frac{\epsilon}{2}OPT$;
    \item $f^q(\gamma) \geq (1-\frac{1}{e})f^q(\gamma*)$, $\gamma*= \argmax_{\gamma'\subset \Gamma, |\gamma'|\leq k}$ $f^q(\gamma')$ (submodularity property of $f^q(.)$);
    \item  $f^q(\gamma*) \geq  f^q(\bar{\gamma})$, $\bar{\gamma}= \argmax_{\gamma'\subset \Gamma, |\gamma'|\leq k} f(\gamma')$
    (Note that, $OPT=f(\bar{\gamma})$)
\end{enumerate}
We can prove with probability $1- \frac{2}{|\Gamma|^l}$ that:
\begin{equation*}
\begin{split}
f(\gamma) &\geq  f^q(\gamma)- \frac{\epsilon}{2}OPT \\
&\geq \Big(1-\frac{1}{e}\Big)f^q(\gamma*)- \frac{\epsilon}{2}OPT \\
&\geq \Big(1-\frac{1}{e}\Big)f^q(\bar{\gamma})- \frac{\epsilon}{2}OPT\\
&\geq \Big(1-\frac{1}{e}\Big)\Big(f(\bar{\gamma})-\frac{\epsilon}{2}OPT\Big)-\frac{\epsilon}{2}OPT\\
& > \Big(1-\frac{1}{e}- \epsilon\Big)OPT \qedhere
\end{split}
\end{equation*} 
%\vspace{1mm}
\end{proof}

%%%%%%%%%%%%%%%%%%%%%%%%%%%%%%%%%% End of Theorem %%%%%%%%%%%
While we are able to achieve a good probabilistic approximation with respect to the optimal value $OPT$, deciding the number of samples is not straightforward. In practice, we do not know the value of $OPT$ beforehand, which affects the number of samples needed. However, notice that $OPT$ is bounded by the number of uncovered pairs $m_u$. Moreover, the number of samples $q(\epsilon/2)$ depends on the ratio  $\frac{m_u}{OPT}$. Thus, increasing this ratio while keeping the quality constant requires more samples. Also, if $OPT$ (which depends on $X$) is close to the number of uncovered pairs $m_u$, we need fewer samples to achieve the mentioned quality. In the experiments, we assume this ratio to be constant. Next, we propose another approximation scheme where we can reduce the number of samples by avoiding the term $OPT$ in the sample size while waiving the assumption involving constants.

Let $M_u$ and $m_u$ be the set and number of uncovered pairs by $X$ respectively in the initial graph. Let us assume, \newline
\begin{equation*}
\bar{q}(\epsilon)\geq \frac{3(l+k)log(|\Gamma|)}{\epsilon^2}
\end{equation*}

\begin{cor}\label{cor:anyE2}
Given $\epsilon$ $(0<\epsilon<1)$, a positive integer $l$, a budget $k$, and a sample of independent uncovered node pairs $Q, |Q|=\bar{q}(\epsilon)$, then:
\begin{equation*}
Pr (|f^q(\gamma)- f(\gamma)| < \epsilon \cdot m_u) \geq 1- 2|\Gamma|^{-l}, \forall \gamma \subset \Gamma, |\gamma| \leq k
\end{equation*}
\end{cor}

The proof is given in the Appendix. Next, we provide an approximation bound by our sampling scheme for at least $\bar{q}(\epsilon/2)= \frac{12(l+k)log(|\Gamma|)}{\epsilon^2}$ samples.

%of $(1-\frac{1}{e})OPT- \epsilon.m_u$ by our sampling scheme if the number of samples is at least $\bar{q}(\epsilon/2)= \frac{12(l+k)log(|\Gamma|)}{\epsilon^2}$.

\begin{cor}\label{cor:BUS_approx2}
  Algorithm \ref{algo:BUS} ensures $f(\gamma)\geq (1-\frac{1}{e})OPT- \epsilon.m_u$ with high probability $(1- \frac{2}{|\Gamma|^l})$ for $\bar{q}(\epsilon/2)$ samples.% where $\gamma$ is its output.
\end{cor}

%{\renewcommand{\arraystretch}{1.5}
\begin{table}[t]
\vspace{-2mm}
\centering
%\small
\begin{tabular}{|c | c | c |}
\hline
\textbf{Thm.}& \textbf{$\#$Samples}& \textbf{Approximations} \\
\hline
\textbf{Thm. \ref{thm:BUS_approx}}& $O(\frac{m_uklog(|\Gamma|)}{\epsilon^2.OPT})$ & $f(\gamma)>(1-\frac{1}{e}- \epsilon)OPT $\\
\hline
\textbf{Cor. \ref{cor:BUS_approx2}}& $O(\frac{klog(|\Gamma|)}{\epsilon^2})$ & $f(\gamma)>(1-\frac{1}{e})OPT- \epsilon.m_u$\\
\hline
\end{tabular}
%\vspace{-2mm}
\caption{$\#$Samples and approximations with prob. $(1- \frac{2}{|\Gamma|^l})$. \label{table:sample_summary}}
%\vspace{-5mm}
 \end{table}
 %}

%%%%%%%%%%%%% The summarizing Table %%%%%%%%%%%%%%
This proof is also in the Appendix. Table \ref{table:sample_summary} summarizes the number of samples and corresponding bounds for Algorithm \ref{algo:BUS}. Theorem \ref{thm:BUS_approx} ensures higher quality with higher number of samples than Corollary \ref{cor:BUS_approx2}. On the other hand, Corollary \ref{cor:BUS_approx2} does not assume anything about the ratio $\frac{m_u}{OPT}$. The results reflect a trade-off between number of samples and accuracy.

 Theorem \ref{thm:BUS_approx} and Corollary \ref{cor:BUS_approx2} assume that a greedy approach achieves a constant-factor approximation of $(1-1/e)$, which holds only for the RCCO problem (see Sections \ref{sec:analysis_const} and \ref{sec:analysis_greedy}). 
 %However, the same approximation cannot be achieved in polynomial-time for the general CCO problem presented in Section \ref{sec:problem_definition}, as shown in Section \ref{sec:general_settings}. 
 As a consequence, in the case of the general problem, the guarantees discussed in this Section apply only for each iteration of our sampling algorithm, but not for the final results. In other words, BUS provides theoretical quality guarantees that each edge selected in an iteration of the algorithm achieves a coverage within bounded distance from the optimal edge. Nonetheless, experimental results show, in practice, BUS is also effective in the general setting.

\section{Experimental Results}
\label{sec:exp}

\begin{table}[t]
\centering
\small
\begin{tabular}{| c | c | c |}
\hline
\textbf{Dataset Name}& \textbf{$|V|$} & \textbf{$|E|$}\\
\hline
\textbf{ca-GrQc (CG)}& 5K & 14K\\
\hline
\textbf{email-Enron (EE)}& 36K & 183K\\
\hline
\textbf{loc-Brightkite (LB)}& 58K & 214K\\
\hline
\textbf{loc-Gowalla (LG)}& 196K & 950K\\
\hline
\textbf{web-Stanford (WS)}& 280K & 2.3M\\
\hline
\textbf{DBLP (DB)}& 1.1M  & 5M\\
\hline
\end{tabular}
%\vspace{-2mm}
\caption{ Dataset description and statistics. \label{table:data_description}}
%\vs
%\vspace{-2mm}
 \end{table}

 \iffalse
 %%%%%%%%%%%%%%%%%%%%%%%%%%%%%%5 upper bound figure
 \begin{figure}[t]
    \centering
%    \vspace{-5mm}
    \subfloat[CG]{
        \includegraphics[keepaspectratio, width=.23\textwidth]{Figures/Experiments/mot_cg.pdf}
       \label{fig:mot1}
    }
    \subfloat[Synthetic]{
       \includegraphics[keepaspectratio, width=.23\textwidth]{Figures/Experiments/mot_syn4k.pdf}
       \label{fig:mot2}
    }
%       \vspace{-3mm}
     \caption{ Improved coverage centrality produced by BUS (our algorithm) and other baselines on the CG and synthetic dataset.}\label{fig:mot}
  %   \vspace{-2mm}
\end{figure}
\fi
 
 %%%%%%%%%%%%%%%%%%%%%%%%%%%%%%%%%%%%%%%%%%%%%%%%%%

%%%%%%%%%%%%%%%%%%%%%%%%%%%%%%%%

\begin{table}[t]
\centering
\small
\begin{tabular}{| c | c | c | c |}
\hline
&\multicolumn{3}{c|}{\textbf{Ratio}} \\
\hline

\textbf{Data}& $k=5$ & $k=10$ & $k=15$  \\
\hline
Co-authorship & $1.02$& $1.14$ & $1.17$  \\
\hline
Synthetic & $1.0$& $1.0$ & $1.0$  \\
\hline
\end{tabular}
%\vspace{-1mm}
\caption{The ratio between the solutions produced by GES for the general (CCO) and constrained (RCCO) settings. \label{table:motivation}}
 %\vspace{-6mm}
 \end{table}

\begin{figure*}[t]
\small
 \vspace{-6mm}
    \centering
    \subfloat[Quality on CG]{\includegraphics[width=0.30\textwidth]{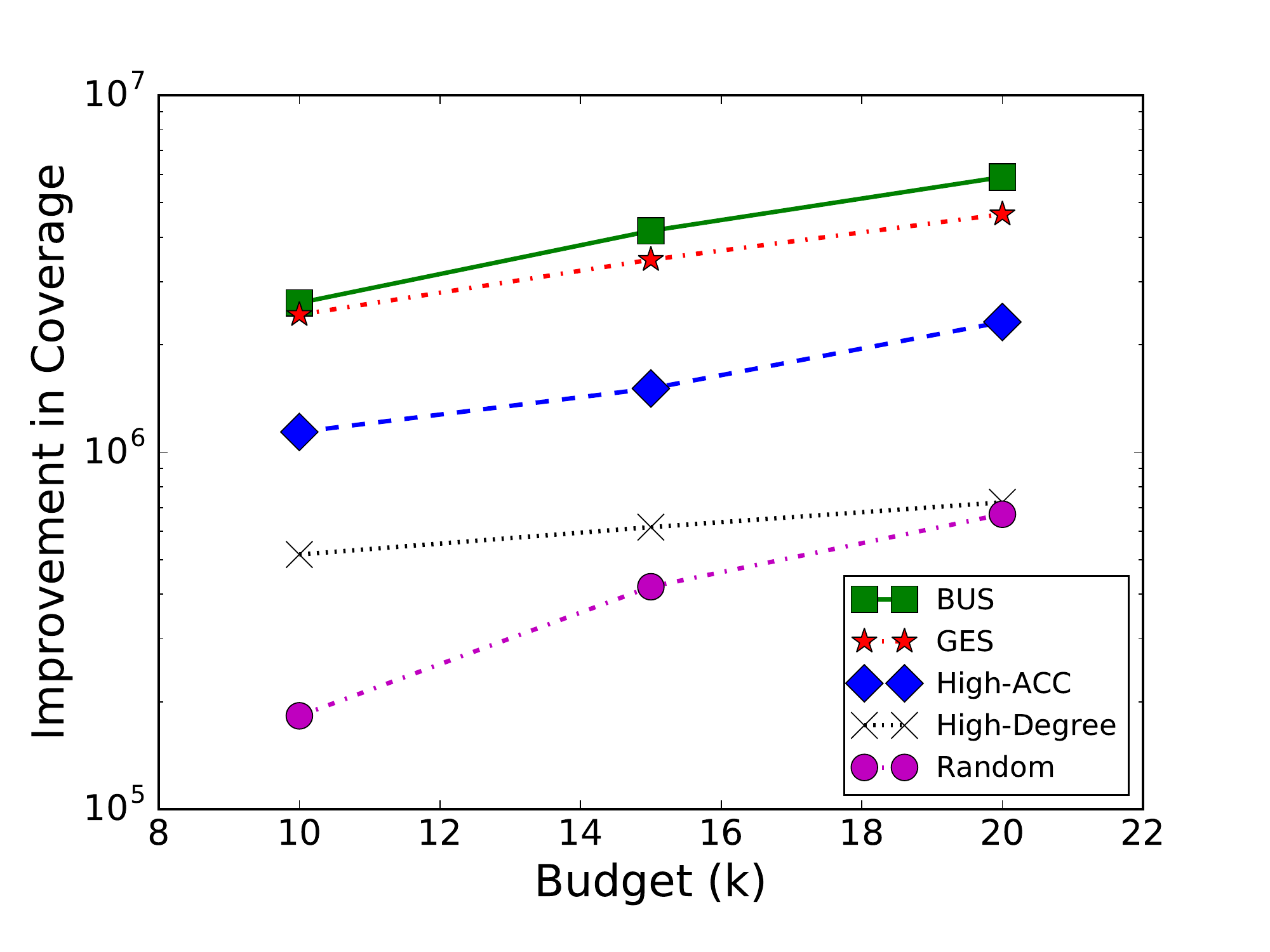}\label{fig:vs_greedy}}
    \subfloat[Fixed Budget]{\includegraphics[width=0.30\textwidth]{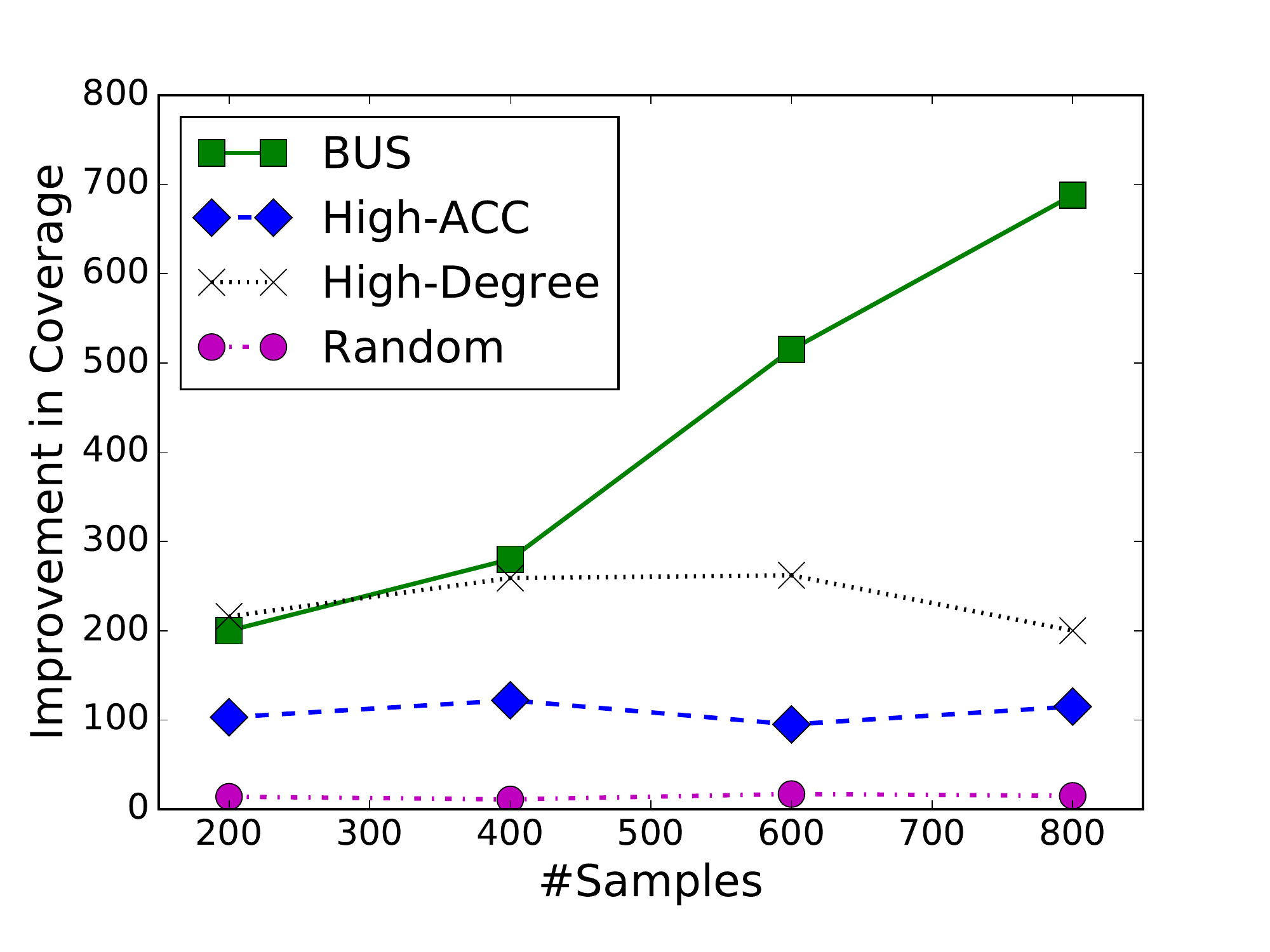}\label{fig:fix_budget}}
    \subfloat[Fixed $\#$Sample]{\includegraphics[width=0.30\textwidth]{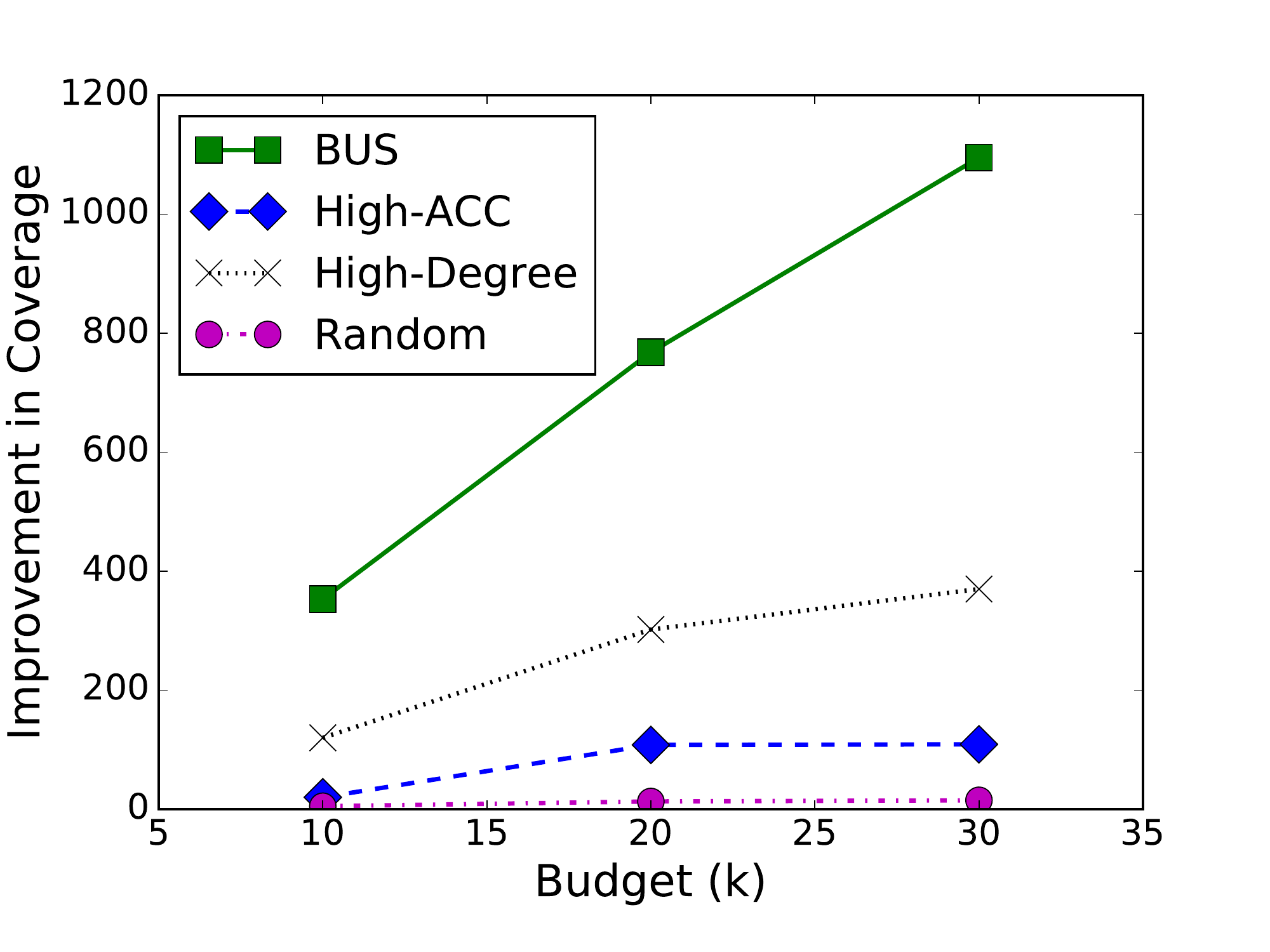}\label{fig:fix_sample}}
 %\vspace{-2mm}
    \caption{ (a) BUS vs. Greedy: Improved coverage centrality produced by different algorithms on the CG dataset. (b-c) Comparison with baselines on the EE dataset varying (b) the number of samples and (c) the budget.  \label{fig:baselines_param}}
%     \vspace{-2mm}
\end{figure*}

 %%%%%%%%%%%%%%%%%%%%%%%%%%%%%%%%%%%%%%%%%%%%

\textbf{Experimental Setup and Data:} 
We evaluate the quality and scalability of our algorithms on real-world networks. All experiments were conducted on a $3.30$GHz Intel Core i7 machine with $30$ GB RAM. Algorithms were implemented in Java and all datasets applied are available online\footnote{\small  Datasets collected from (1) \url{https://snap.stanford.edu/data/index.html}, (2) \url{http://dblp.uni-trier.de}, and (3) \url{http://www-personal.umich.edu/~mejn/netdata/}}. Table \ref{table:data_description} shows dataset statistics. The graphs are undirected and we consider the largest connected component for our experiments. Results reported are averages of $10$ repetitions.

%The number of samples involve parameters $\Gamma$, $k$, $\frac{m_u}{OPT}$, and $\epsilon$. 
We set the candidate of edges $\Gamma$ as those edges from $X$ to the remaining vertices that are absent in the initial graph (i.e. $\Gamma = \{(u,v)|u \in X \land v \in V\setminus X \land (u,v) \notin E\}$). %Of course, the reduction in number of candidate edges will result in lesser running time. We also present the budget $k$ and $\epsilon$ (the error) whenever needed. 
%We assume $\frac{m_u}{OPT}=1$ unless mentioned otherwise.
The set of target nodes $X$ is randomly selected from the set of all nodes.

\textbf{Baselines:} We consider three baselines in our experiments: 1) \textbf{High-ACC:} Applies \textit{maximum adaptive centrality coverage} \cite{yoshida2014,mahmoodye2016} and adds edges between target nodes $X$ and the top-$k$ centrality set; 2) \textbf{High-Degree:} Selects edges between the target nodes $X$ and the top $k$ high degree nodes; 3) \textbf{Random:} Randomly chooses $k$ edges from $\Gamma$ which are not present in the graph. We also compare our sampling algorithm (BUS) against our Greedy solution (GES) and show that BUS is more efficient while producing  similar results.

 \subsection{GES: RCCO vs CCO}
 \label{sec:exp_rcco_vs_cco}
We compare coverage centrality optimization (CCO) and its restricted version (RCCO) empirically by applying 
%We have mentioned earlier (Section \ref{sec:analysis_const}) that solutions produced by GES in RCCO and CCO are similar. 
 GES to two small datasets: a co-authorship (NetScience) and a synthetic (Barabasi) network. The target set size $|X|$ is set to $5$. 
 Table \ref{table:motivation} shows the ratio between results for CCO and RCCO varying the budget $k$. The results, close to $1$, provide evidence that RCCO is based on realistic assumptions.

%%%%%%%%%%%%%%%%%%%%%%%%%%%%%%%%%%%%vs GREEDY%%%%%%%%%%%%%%%%%%%%%%%%%%%%%%

\subsection{BUS vs. GES}
We apply only the smallest dataset (CG) in this experiment, as the GES algorithm is not scalable and computing all-pair-shortest-paths is required. For BUS, we set the error $\epsilon=0.3$. 
First, %we compare our sampling scheme BUS with Greedy (GES) in order to 
we evaluate the effect of sampling on quality, which we theoretically analyzed in Theorem \ref{thm:BUS_approx} and Corollary \ref{cor:BUS_approx2}.

 Fig. \ref{fig:vs_greedy} shows the number of new pairs covered by the algorithms. Table \ref{table:vs_greedy} shows the running times and the quality of BUS relative to the baselines---i.e. how many times more pairs are covered by BUS compared to a given baseline. BUS and GES produce results at least $2$ times better than the baselines. Moreover, BUS achieves results comparable to GES while being $2$-$3$ orders of magnitude faster. 
%Interestingly, only adding $10$ edges result into a high number ($>10^6$) of pairs coverage by our method, BUS. 

\iffalse
\begin{figure}[ht]
    \centering
   \includegraphics[width=0.45\textwidth]{Figures/Experiments/greedy.pdf}
 %\vspace{-2mm}
    \caption{ Comparison on CG: the actual number of pairs covered by different algorithm after the budget edges are added. \label{fig:vs_greedy}}
 \end{figure}

\fi

%%%%%%%%%%%%%%%%%%%%%%%%%%%%%%%%%%%%%%%%%%%%%%%%%%%%%%%%%%%%%%%%%%%%%%%%%%%%%%%%%%%

%%%%%%%%%%%%%%%%%%%%%%%%%%%%%%%%%%%%%%%%%%%%%%%%%%%%%%%%%%%%%%%%%%%%%%%%%%%%%%%%%%%

\subsection{Results for Large Graphs}
%\todo{add dblp} 
We compare our sampling-based algorithm against the baseline methods using large graphs (EE, LB, LG, WS and DB). Due to the high cost of computing all-pairs shortest-paths, we estimate the coverage centrality based on $10K$ randomly selected pairs. For High-ACC, we also use sampling for adaptive coverage centrality computation~\cite{yoshida2014,mahmoodye2016} and the same number of samples is used by High-ACC and BUS. The budget and target set size are set as $20$ and $5$, respectively. 

Table \ref{table:vs_baseline} shows the results, where the quality is relative to BUS results. BUS takes a few minutes ($8,15, 17, 45,85$ minutes for EE, LB, WS, LG and DB respectively) to run and significantly outperforms the baselines. This is due to the fact that existing approaches do not take into account the dependencies between the edges selected in the coverage centrality. BUS selects the edges sequentially, considering the effect of edges selected in previous steps.

%%%%%%%%%%%%%%%%%%%%%%%%%%%%%%%%%%%%%%%%%%%%%%%%%%%%%%%%%%%%%%%%%%%%%%%%%%%%%%%%%%%

%%%%%%%%%%%%%%%%%%%%%%%%%%%%%%%%%%%%%%%%%%%%%%%%%%%%%%%%%%%%%%%%%%%%%%%%%%%%%%%%%%%

\begin{table*}[ht]
\centering
\small
\begin{tabular}{| c | c | c | c | c | c | c |  c | c|}
\hline
&\multicolumn{4}{c|}{\textbf{Coverage of BUS (relative to baselines)}} & \multicolumn{3}{c|}{\textbf{Time [sec.]}} & \textbf{$\#$ Samples}\\
\hline

\textbf{Budget}& GES & HIgh-ACC & High-Degree & Random  & GES & High-ACC & BUS & BUS\\
\hline
$k=10$ & $1.08$& $2.46$ & $5.41$ & $14.45$ & $>7200$& $157.1$ & $5.1$ & $2560$\\
\hline
$k=15$ & $1.21$& $2.92$& $7.29$& $9.98$ &$>7200$ & $156.9$ & $10.1$ & $3840$\\
\hline
$k=20$ & $1.29$ &$2.78$ & $9.96$& $9.59$& $>7200$& $157.2$ & $18.2$ & $5120$\\
\hline
\end{tabular}
%\vspace{-2mm}
\caption{Comparison between our sampling algorithm (BUS) and the baselines, including our Greedy (GES) approach, using the CG dataset and varying the budget $k$. We evaluate the coverage of BUS relative to the baselines---i.e. how many times more new pairs are covered by BUS compared to the baseline. \label{table:vs_greedy}}
 %\vspace{-3mm}
 \end{table*}

 \begin{table}[ht]
\centering
\scriptsize
\begin{tabular}{| c | c | c | c | c | c |  c | c|}
\hline
&\multicolumn{3}{c|}{\textbf{Coverage of BUS (relative to baselines)}} & \textbf{$\#$ Samples}\\
\hline

\textbf{Data} & High-Acc & High-Degree & Random & BUS\\
\hline
EE &  $4.88$ & $2.74$ & $51$ & $6462$\\
\hline
LB & $3.3$& $2.3$& $33.8$  & $6796$\\
\hline
LG & $3.3$ & $4.2$& $62$ & $4255$\\
\hline
WS & $1.89$ & $1.95$& $4.8$ & $2000$\\
\hline
DB & $2.5$ & $1.6$& $5$ & $875$\\
\hline
\end{tabular}
%\vspace{-3mm}
\caption{BUS vs. baselines for large datasets.  \label{table:vs_baseline}}
%\vspace{-2mm}
 \end{table}

%%%%%%%%%%%%%%%%%%%%%%%%%%%%%%%%%%%%%%%%%%%%%%%%%%%%%%%%%%%%%%%%%%%%%%%%%%%%%%%%%%%
\subsection{Parameter Sensitivity}
The main parameters of BUS are the budget and the number of samples---both affect the error $\epsilon$, as discussed in Thm. \ref{thm:BUS_approx} and Cor. \ref{cor:BUS_approx2}. We study the impact of these two parameters on performance. % We experiment with these two parameters and study their impact on performance. 
Again, we estimate coverage using $10K$ randomly selected pairs of nodes. %First, we fix the budget and vary the number of samples. %Note that the error $\epsilon$ depends on both the budget $k$ and the number of samples (see Thm. \ref{thm:BUS_approx} and Cor. \ref{cor:BUS_approx2}). %Here we want to explore the quality of our algorithm separating these two parameters. 

Figure \ref{fig:fix_budget} shows the results on EE data for budget $20$ and target set size $5$. With $600$ samples, BUS produces results at least $2$ times better than the baselines. Next, we fix the number of samples and vary the budget. Figure \ref{fig:fix_sample} shows the results on EE data with $10K$ samples and $5$ target nodes. BUS produces results at least $2.5$ times better than the baselines. %The results are even better as the budget grows. 
Moreover, BUS takes only $30$ seconds to run with budget of $30$ and $1000$ samples. We find that the running time grows linearly with the budget for a fixed number of samples. These results validate the running time analysis from Section \ref{sec::approximation_algorithms}.

%%%%%%%%%%%%%%%%%%%%%%%%%%%%%%%%%%%%%%%%%%%%%%%%%%%%%%%%%%%%%%%%%%%%%%%%%%%%%%%%%%%

\begin{table}[t]
\centering
\small
\begin{tabular}{| c | c | c | c | c | c |  c | c| c| c|}
\hline
&\multicolumn{3}{c|}{Influence} & \multicolumn{3}{c|}{Distance} & \multicolumn{3}{c|}{Closeness}\\
\hline
\textbf{$k$} & EE & LB & LG & EE & LB & LG & EE & LB & LG\\
\hline
$25$ & $57.7$ & $12.2$ & $10.7$ & $2.7$ & $1.2$ & $2.2$ & $2.0$& $2.0$ & $1.0$ \\
\hline
$50$ & $96.8$ & $17.5$ & $92.7$ & $3.8$ & $3.5$ & $3.3$ & $4.9$ &$3.9$ & $4.0$ \\
\hline
$75$ & $134.3$ & $29.1$ & $45.9$ & $5.2$ & $2.1$ & $2.3$ & $5.9$ & $2.3$ & $1.9$ \\
\hline
\end{tabular}
%\vspace{-2mm}
\caption{ Improvement of other metrics after adding the edges found by BUS: the numbers are improvement in percentage with respect to the value for the initial graph. \label{table:other_metric}}
%\vspace{-5mm}
 \end{table}

\subsection{Impact on other Metrics}

While this paper is focused on optimizing Coverage Centrality, it is interesting to analyze how our methods affect other relevant metrics. Here, we look at the following ones:
1) influence, 2) average shortest-path distance, and 3) closeness centrality. The idea is to assess how BUS improves the influence of the target nodes, decreases the distances from the target to the remaining nodes, and increases the closeness centrality of these nodes as new edges are added to the graph. For influence analysis, we consider the popular independent cascade model \cite{kempe2003maximizing} assuming edge probabilities as $0.1$. In all the experiments, we fix the number of sampled pairs at $1000$ and choose $10$ nodes, uniformly at random, as the target set $X$. The metrics are computed before and after the addition of edges and presented as the relative improvement in percentage. Notice that because target nodes are chosen at random, increasing the budget does not necessarily lead to an increase in the metrics considered. 

Results are presented in Table \ref{table:other_metric}. There is a significant improvement of the three metrics as the budget ($k$) increases. For influence, the number of seed nodes is small, and thus the relative improvement for increasing $k$ is large. The improvement of the other metrics is also significant. For instance, in EE, the decrease in distance is nearly $5\%$, which is approximately $72K$, for a budget of $75$. %The initial distance was nearly $4.1*10*36000$ ($10$ is size of the target set and $36K$ is size of EE). So, the actual improvement in distance is $0.2*10*36000 = 72K$ which is a significant amount. 
%%%%%%%%%%%%%%%%%%%%%%%%%%%%%%%%%%%%%%%%%%%%%%%%%%%%%%%%%%%%%%%%%%%%%%%%%%%%%%%%%%%%%%%%%%%%

\iffalse
{\color{red}
     maintenance over dynamic graphs,
     DBLP vs Barabasi,
    can it mimic?
     why is it useful?
     How to define robustness,
     Do mid-central nodes become high central? Centrality spectrum, distribution.
     }
 
\fi

\section{Towards More General Settings}
\label{sec:general_settings}

 We start by extending our approaches to solve the Coverage Centrality optimization problem on \textit{directed} graphs (e.g. our motivating example in Figure \ref{fig:motivation}). %We do not propose any new algorithm for directed graphs. Instead we find the appropriate setting to apply our proposed solutions. %We define the CCO problem for directed graphs in similar way. 
 In this setting, edges are added from or towards the target nodes $X$%\todo{here} %. We assume that the candidate sets of edges, $\Gamma$ includes the edges which have direction towards the target nodes.  
 ---i.e. directed edges in $\Gamma$ are of the form $(u,v)$ or $(v,u)$ where $u \in V\setminus X, v\in X$. %The direction is from $u$ to $v$. %\todo{does not match with pervious}
 
\begin{problem}\textbf{Coverage Centrality Optimization in Directed Graphs (CCO-D):} Given a directed network $G=(V,E)$, a node set $X \subset V$, a candidate edge set $\Gamma$, and a budget $k$, find edges $E_s\subset \Gamma$, such that $|E_s|\leq k$ and $C_m(X)$ is maximized.
\label{def:CCO-D}
\end{problem}

We assume the same constraint $S^2$ for CCO-D. The next theorem shows an inapproximability result for this problem.
%that CCO-D is also NP-hard to approximate within a factor greater than $(1- \frac{1}{e})$.% and thus the same greedy algorithm described in Section~\ref{sec:greedy} can be applied with approximation guarantees.

\begin{thm}\label{thm:directed_approx_lower} CCO-D under $S^2$ cannot be approximated within a factor greater than $(1- \frac{1}{e})$.
\end{thm}
 
 \begin{proof}
  We give a $L$-reduction \cite{williamson2011design} from the maximum coverage (MSC) problem with parameters $x$ and $y$. Our reduction is such that following two equations are satisfied:
 \begin{equation}
     OPT(I_{CD }) \leq xOPT(I_{MSC})
     \vspace{-2mm}
 \end{equation}
  \begin{equation}
     OPT(I_{MSC})-s(T^M) \leq y(OPT(I_{CD })-s(T^C))
 \end{equation}
where $I_{MSC}$ and $I_{CD }$ are the two problem instances, $OPT$ denotes the optimal values of the optimization problem instances. $s(T^M)$ and $s(T^C)$ denote any solution of the MSC and CCO-D instances respectively. If the conditions hold and CCO-D has an $\alpha$ approximation, then MSC has an $(1-xy(1-\alpha))$ approximation algorithm. However, MSC is NP-hard to approximate within a factor greater than $(1-\frac{1}{e})$. It follows that $(1-xy(1-\alpha))< (1-\frac{1}{e})$, or, $\alpha < (1-\frac{1}{xye})$ \cite{crescenzi2015}. So, if the above two conditions are satisfied then CCO-D is NP-hard to approximate within a factor greater than $(1-\frac{1}{xye})$. 

 Consider an instance of the Maximum Coverage (MSC) problem, defined by a collection of subsets $S_{1},S_{2},...,S_{m}$ for a universal set of items $U=\{ u_{1},u_{2},...,u_{n} \}$. To define a corresponding CCO-D instance, we construct an directed graph with $m+n+3$ nodes: there are nodes $i$ and $j$ corresponding to each set $S_{i}$ and each element $u_{j}$ respectively, and an directed edge $(i,j)$ whenever $u_{j}\in S_{i}$. Three more nodes ($a,b$ and $c$) are added to the graph where $a$ is in $X$% (the given set of nodes in CCO)
. Node $c$ is connected to $u_{i}$ by $(c,u_i)$ for all $i \in {1,2,...,n}$. Node $b$ is attached to $a$ by $(b,a)$ and $c$ by $(b,c)$. Figure \ref{fig:directed_inapprox} shows an example of this construction. The reduction clearly takes polynomial time. The candidate set $\Gamma$ consists of the edges between $a$ and the set $S$. 
For CCO-D, the set to be covered, $Z$ contains pairs in the form $(b,u)$ where $u \in U$.
 
 Let the solution of $I_{CD}$ be $s(T^C)$. The centrality of $a$ will increase by $s(T^C)$ to cover the pairs in $Q$. Note that $s(T^C)= s(T^M)$ by  construction. It follows that both the conditions are satisfied when $x=y=1$. So, CCO-D is NP-hard to approximate within a factor grater than $(1-\frac{1}{e})$.
 \end{proof}

\begin{figure}[t]
    \centering
      \vspace{-3mm}
   \includegraphics[width=.35\textwidth]{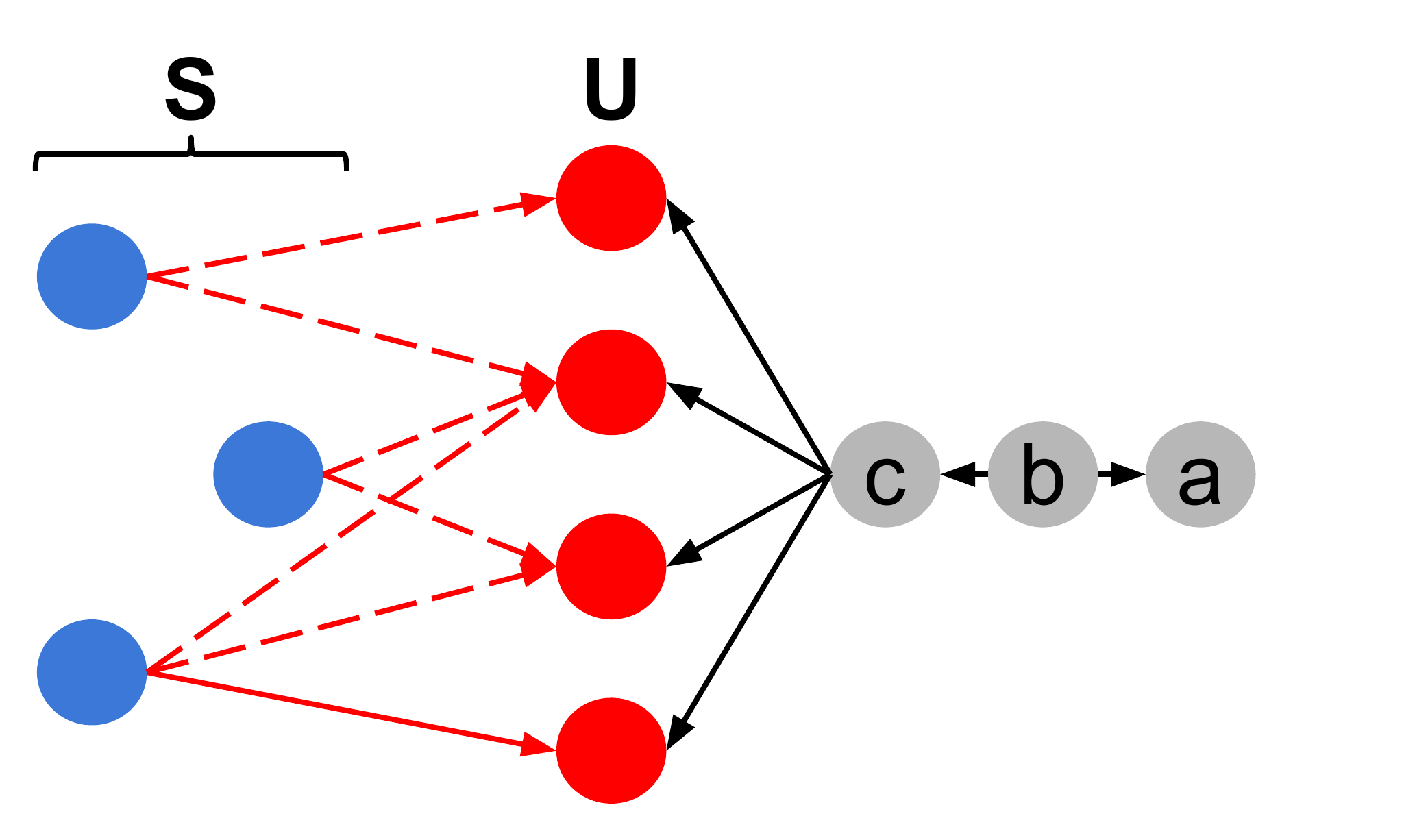}
   \vspace{-1mm}
    \caption{Example of reduction from Maximum Set Coverage%: the red and blue nodes belong to set $U$ and $S$ respectively 
    , where $|U|=4$ and $|S|=3$. % The red edges are when the element belongs to a particular subset as in Set Cover problem. The black edges are further added as per construction. 
     Target set $X=\{a\}$ and candidate edges $\Gamma$ connect $a$ to nodes in set $S$.\label{fig:directed_inapprox}}
%\vspace{-1mm}
 \end{figure}
The next theorem shows that the objective function associated with Problem~\ref{def:CCO-D} under $S^2$ is also monotone and submodular, as was the case for the undirected setting.% and thus the same greedy algorithm described in Section~\ref{sec:greedy} can be applied with approximation guarantees.

\begin{thm}\label{thm:submodular_wt} Given $X$, the objective function, $f(E_s)=C_m(X)$ in CCO-D is monotone and submodular.
\end{thm}

The proof for Theorem~\ref{thm:submodular_wt} is similar to that for Theorem~\ref{thm:submodular}. Based on this Theorem, our algorithm (BUS) can be applied to solve CCO-D with similar guarantees. In other words, our approach is agnostic to the direction of edges. Interestingly, Theorem \ref{thm:directed_approx_lower} and Theorem \ref{thm:submodular_wt} certify that GES achieves the best approximation for the constrained CCO-D problem.

We also briefly discuss group centrality optimization under different settings. In particular, we focus on possible restrictions on the set of candidate edges $\Gamma$. \textit{For undirected graphs}, \textbf{$S^0$:} $\Gamma$ is a subset of the set of absent edges, \textbf{$S^1$:} $\Gamma$ consists of absent edges of the form $(u,v)$ where either $u$ or $v$ belongs to the target set $X$, and \textbf{$S^2$:} a pair is covered using at most one newly added edge. For \textit{directed graphs}, \textbf{$S^3$:} $\Gamma$ is a subset of the set of absent edges with arbitrary direction and \textbf{$S^4$:} $\Gamma$ consists of absent edges of the form $(u,v)$ where either $u$ or $v$ belongs to $X$,  with any direction.
  The hardness of these problems can be assessed with variations of the reasoning applied in Theorem $\ref{thm:nphard}$. Table \ref{table:problems} summarizes the different problem settings. % we have investigated. 
  We have already proven that the objective function is submodular for $S^1$ and $S^2$ (undirected) and for $S^4$ and $S^2$ (directed). Additionally, we prove that the objective functions for $S^1$ and $S^4$ individually are not submodular.% and thus all claims in Table \ref{table:problems} are proved.

 \begin{table}[ht]
\centering
%\vspace{-2mm}
\small
\begin{tabular}{| c | c | c| c| c| c| c|}
\hline
& \multicolumn{3}{c|}{\textbf{Undirected}} & \multicolumn{3}{c|}{\textbf{Directed}}\\
\hline
\textbf{Settings}& $S^0$ & $S^1$ & $S^1$, $S^2$ & $S^3$ & $S^4$ & $S^4$,$S^2$\\
\hline
\textbf{Submodularity}& No & No & Yes & No & No & Yes \\
\hline
\end{tabular}
%\vspace{-2mm}
\caption{ Summary of CCO under different settings. \label{table:problems}}
%\vspace{-3mm}
 \end{table}

\noindent
\textbf{Non-submodularity under $S^1$ and $S^4$:} Counterexamples are shown in Fig. \ref{fig:non_submodular}. %to prove the functions are not submodular. \newline
\textit{For $S^1$:} Consider $T=\{(x,a)\}, S=\{\}, e=(x,b)$ and a target node $x$. Clearly, $S\subset T$ and $f(T)=f(S)=0$. But $f(T\cup \{e\})=1$ as $x$ is covering $(a,b)$, whereas $f(S\cup \{e\})=0$. So, $f(T \cup \{e\})- f(T) > f(S \cup \{e\})- f(S)$, and, $f$ is not submodular. 
\textit{For $S^4$:} The proof is similar to $S^1$. Let $T=\{(a,x)\}, S=\{\}, e=(x,b)$ and the target node be $x$. Thus, $S\subset T$ and $f(T)=f(S)=0$. But $f(T\cup \{e\})=1$ as $x$ is covering $(a,b)$, whereas $f(S\cup \{e\})=0$. So, $f(T \cup \{e\})- f(T) > f(S \cup \{e\})- f(S)$ and $f$ is not submodular.

\begin{figure}[ht]
    \centering
    \vspace{-9mm}
    \subfloat[Undirected, $S^1$]{
        \includegraphics[keepaspectratio, width=.23\textwidth]{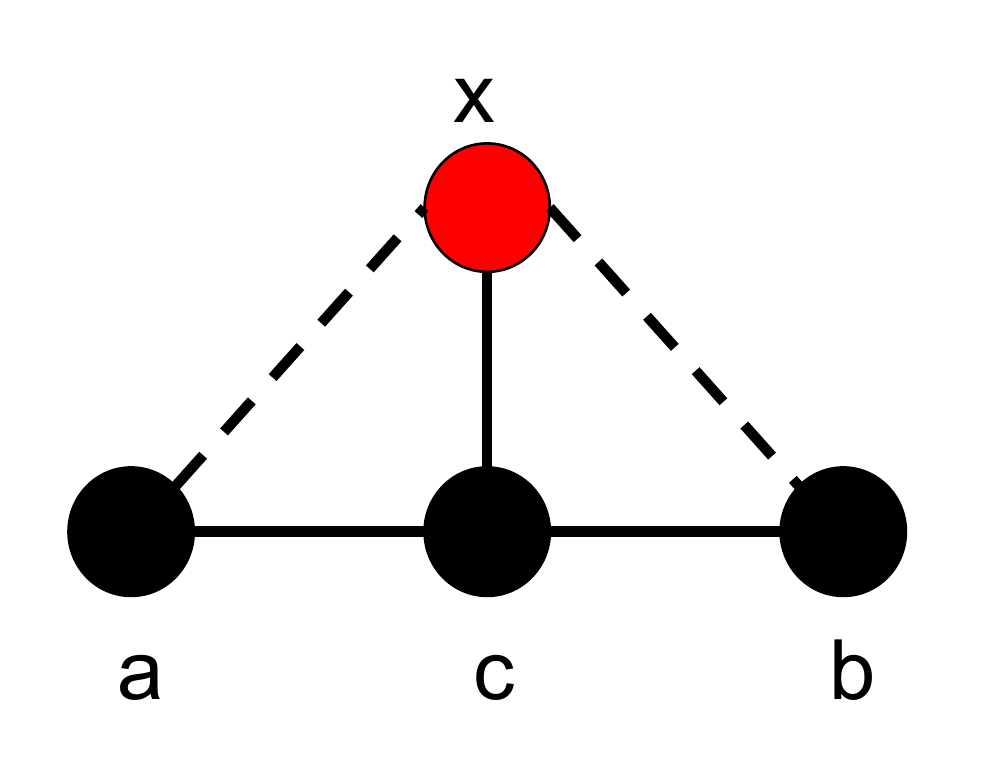}
       \label{fig:S1}
    }
    \subfloat[Directed, $S^4$]{
       \includegraphics[keepaspectratio, width=.23\textwidth]{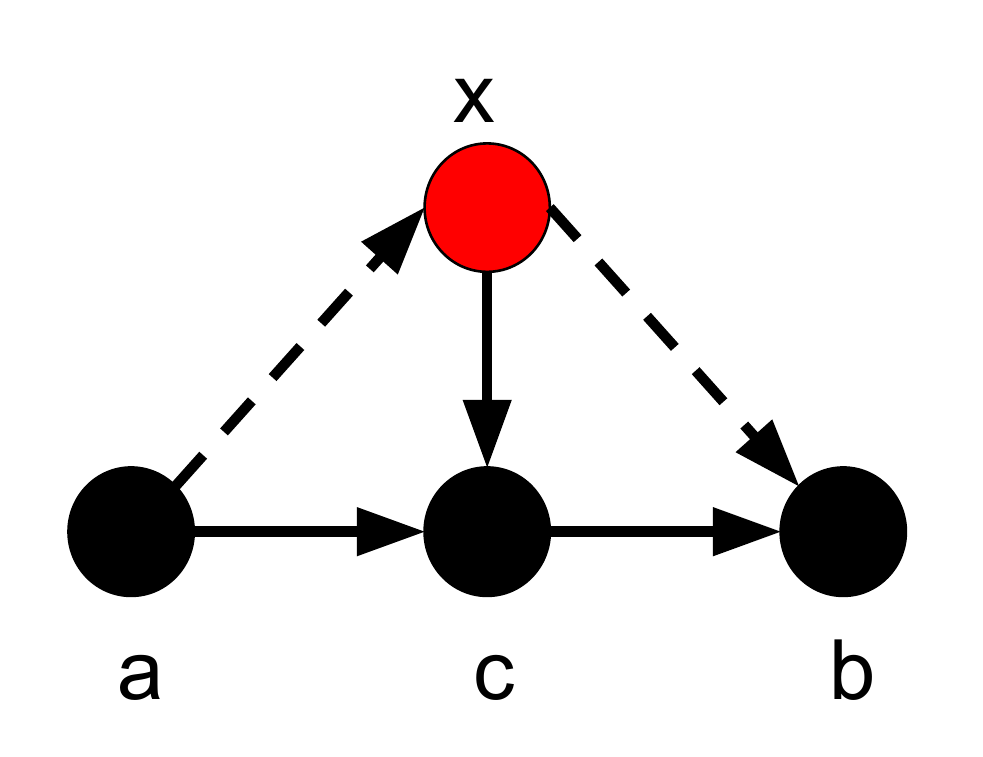}
       \label{fig:S6}
    }
       \vspace{-1mm}
     \caption{Non-submodularity under (a) $S^1$ and (b) $S^4$.}\label{fig:non_submodular}
     \vspace{-4mm}
\end{figure}

%\textbf{Discussion about $S2$ and $S3$: } We have argued that $S2$ is more realistic than $S1$. Though we provide the motivation about choosing $S2$ and $S3$ in Sec. \ref{sec:problem_definition}, we discuss it further here. Note that, our algorithms GES and BUS, are based on uncovered pairs. However, the graph is changing after the edge addition and the already covered paths might change and might not respect the constraint $S3$, but it does not affect our algorithms or our theoretical results because of the following: let $(u,v)$ is covered by our algorithm at a particular step. Now, shortest path between $(u,v)$ might change as graph is changing due to edge addition. But it will be covered as the paths can only change via the edges whose one of the end points is in $X$. 

\iffalse
\textbf{Non optimal behaviour of Greedy}

%Example of Greedy's non optimal behaviour
The submodularity property of the objective function ensures that Greedy algorithm produces a constant time approximation.

\begin{figure}[b]
    \centering
    \subfloat[Optimal]{\includegraphics[width=0.25\textwidth]{Figures/Greedy_O_NO/Greedy_nonoptimal.pdf}\label{fig:greedy_o}}
    \subfloat[Non Optimal Greedy]{\includegraphics[width=0.25\textwidth]{Figures/Greedy_O_NO/optimal.pdf}\label{fig:greedy_no}}
    \caption{Illustration of Greedy's non-optimal behaviour \label{fig:baselines_small_large}}
\end{figure}

\fi

\section{Previous Work}
\label{prev_work}
There is a considerable amount of work on network design targeting various objectives by modifying the network structure and node/edge attributes. %These problems differ mostly on the upgrading models and objective functions considered. % In this section, we cover existing work on network design via optimization. In particular, we describe prior work which aim to optimize different metrics related to shortest paths. 

\textit{General network design problems:} A set of design problems were introduced by Paik et al.~\cite{paik1995}. They focused on vertex upgrades to improve the delays on adjacent edges. %They solve one of the problems proposed for special class of graphs. 
 Krumke et al.~\cite{krumke1998} generalized this model and proposed minimizing the cost of the minimum spanning tree with varying upgrade costs for vertices/edges.
Lin et al.~\cite{lin2015} also proposed a shortest path optimization problem via improving edge weights under a budget constraint and with undirected edges. In~\cite{dilkina2011, medya2016}, the authors studied a different version of the problem, where weights are set to the nodes.

%\textbf{Structural network design: } 
 \textit{Design problems via edge addition:} %Prior work use \emph{edge addition} to improve  global network objectives such as vertex eccentricity, diameter, all-pairs and single-source shortest paths~\cite{meyerson2009,papagelis2011,parotisidis2015selecting,demaine2010,perumal2013}. 
 Meyerson et al.~\cite{meyerson2009} proposed approximation algorithms for single-source and all-pair shortest paths minimization. Faster algorithms for the same problems were presented in~\cite{parotisidis2015selecting}. Demaine et al.~\cite{demaine2010} minimized the diameter of a network and node eccentricity by adding shortcut edges with a constant factor approximation algorithm. Past research had also considered eccentricity minimization in a composite network% where a social node connectivity is improved by additional communication network edges
 ~\cite{perumal2013}. However, all aforementioned problems are based on improving distances
  and hence are complementary to our objective. 

\textit{Centrality computation and related optimization problems:} %Efficient computation of betweenness centrality is very relevant to our problem. 
%These papers are 
This line of research is the most related to the present work.
The first efficient algorithm for betweenness centrality computation was proposed by Brandes \cite{brandes2001}. Recently, \cite{riondato2014} introduced an approach for computing the top-$k$ nodes in terms of betweenness centrality via VC-dimension theory. Yoshida \cite{yoshida2014} studied similar problems ---for both betweenness and coverage centrality--- in the adaptive setting, where shortest paths already covered by selected nodes are not taken into account. Yoshida's algorithm was later improved using a different sampling scheme~\cite{mahmoodye2016}.  Here, we focus on the design version of the problem, where the goal is to optimize the coverage centrality of a target set of nodes by adding edges. When the target set has size one, optimization of different centralities was studied in~\cite{crescenzi2015,ishakian2012framework}. %Surprisingly, optimizing closeness centrality has a constant time approximation but optimizing bewteenness centrality of a node is proven to be NP-hard to approximate within a constant factor for weighted graphs. 
In \cite{parotsidis2016centrality}, the authors solved a similar problem, which is maximizing the expected decrease in the sum of the shortest paths from a single source to the remaining nodes via edge addition.

\section{Conclusions}

In this paper, we studied several variations of a novel network design problem, the group centrality optimization. This general problem has applications in a variety of domains including social, collaboration, and communication networks. From a computational hardness perspective, we have shown that the variations of problem are NP-hard as well as APX-hard. Moreover, we have proposed a simple greedy algorithm, and even faster sampling algorithms, for group centrality optimization. Our algorithms provide theoretical quality guarantees under realistic constrained versions of the problem and also outperform the baseline methods by up $5$ times in real datasets. While we have focused our discussion on coverage centrality, our results also generalize to betweenness centrality. From a broader point of view, we believe that this paper highlights interesting properties of network design problems compared to their standard search counterparts.

%. We further studied the problem under restricted yet realistic settings. In the restricted setting, the problem is APX-hard but our proposed greedy approach achieves a nearly optimal approximation factor. The randomized technique also obtains a probabilistic approximation guarantee. %The randomized algorithm is very efficient that scale to large million-node instances and consistently outperform alternatives. 
%More general settings of the problem were also discussed. 
%We evaluated our approaches on several real-world graphs showing that it outperforms the best baseline solution in terms of quality---coverage centrality achieved---by up to $5$ times. %Moreover, 
%the proposed randomized algorithm is very efficient, scaling to large (million-node) instances.% while achieving similar results to our original greedy scheme in terms of quality.

As future work, we will investigate the dynamic version of the problem~\cite{hayashi2015fully,lerman2010centrality,takaguchi2016coverage}, where coverage centrality has to be maintained under temporal, and possibly adversarial, edge updates. This problem has interesting connections with existing work on \textit{Game Theory}~\cite{ciftcioglu2016}. Moreover, we will study other design problems that optimize social influence and consensus in networks~\cite{Khalil2014,chaoji2012recommendations}.

%\input{7_1_appendix}
%\input{7_appendix}

\bibliographystyle{abbrv}
\bibliography{icdm2017}  % vldb_sample.bib is the name of the Bibliography in this case
%\section{Appendix}

\section*{Appendix}

\noindent
\textbf{Proof of Corollary \ref{cor:anyE2}} \newline
%\begin{proof}
Using Lemmas \ref{lemma:expectation} and \ref{lemma:anyE}:
\begin{equation*}
\begin{split}
Pr (|f^q(\gamma)- f(\gamma)| &\geq \delta \cdot f(\gamma))\\
%$=Pr (|\frac{q}{M_u}f^q(\gamma)- \frac{q}{M_u}f(\gamma)| \geq \frac{q}{M_u}.\delta.f(\gamma))$ 
%\newline
%$=Pr (|g^q(\gamma)- \frac{q}{M_u}f(\gamma)| \geq \frac{q}{M_u}.\delta f(\gamma))$ 
%\newline
Pr (|g^q(\gamma)- E(g^q(\gamma))| &\geq \delta E(g^q(\gamma))) 
\end{split}
\end{equation*}
The rest of the proof follows that for Lemma \ref{lemma:anyE} but replacing $OPT$ by $m_u$. As the samples are independent, we can apply the Chernoff bound:
\begin{equation*}
Pr \Big(|g^q(\gamma)- \frac{\bar{q}}{m_u}f(\gamma)| \geq \frac{\bar{q}}{m_u}\delta f(\gamma)\Big) \leq 2\exp \Big(-\frac{\delta^2}{3} \frac{\bar{q}}{m_u}f(\gamma)\Big)
\end{equation*}
Now, substituting $\delta=\frac{\epsilon \cdot m_u}{f(\gamma)}$ and $\bar{q}$:
\begin{equation*}
Pr (|f^q(\gamma)- f(\gamma)| \geq \epsilon \cdot m_u) \leq 2\exp \Big(-\frac{m_u}{f(\gamma)}(l+k)log(\Gamma)\Big)
\end{equation*}

Using the fact that $m_u\geq f(\gamma)$:
\begin{equation*}
Pr (|f^q(\gamma)- f(\gamma)| \geq \epsilon m_u) \leq 2|\Gamma|^{-(l+k)}
\end{equation*}
Now, we apply the union bound over all possible size-$k$ subsets of $\gamma \subset \Gamma$ (there are $|\Gamma|^k$) to get the following:

\begin{equation*}
Pr (|f^q(\gamma)- f(\gamma)| \geq \epsilon \cdot m_u) < 2|\Gamma|^{-l}, \forall \gamma \subset \Gamma
\end{equation*}
\begin{equation*}
Pr (|f^q(\gamma)- f(\gamma)| < \epsilon \cdot m_u) \geq 1- 2|\Gamma|^{-l}, \forall \gamma \subset \Gamma
\end{equation*}
%\end{proof}

\vspace{1mm}
\noindent
\textbf{Proof of Corollary \ref{cor:BUS_approx2}} \newline
%\begin{proof}
By the same arguments as in the proof of Theorem \ref{thm:BUS_approx}, with probability $1- \frac{2}{|\Gamma|^l}$:
\begin{equation*}
\begin{split}
f(\gamma) &\geq  f^q(\gamma)- \frac{\epsilon}{2}m_u \\
& \geq \Big(1-\frac{1}{e}\Big)f^q(\gamma*)- \frac{\epsilon}{2}m_u \\
& \geq \Big(1-\frac{1}{e}\Big)f^q(\bar{\gamma})- \frac{\epsilon}{2}m_u \\
& \geq \Big(1-\frac{1}{e}\Big)\Big(f(\bar{\gamma})-\frac{\epsilon}{2}M_u\Big)-\frac{\epsilon}{2}m_u \\ 
& = \Big(1-\frac{1}{e}\Big)OPT- \Big(\epsilon-\frac{\epsilon}{2e}\Big) m_u \\ 
& > \Big(1-\frac{1}{e}\Big)OPT- \epsilon m_u 
\end{split}
\end{equation*}
%\end{proof}

 %%%%%%%%%%%%%%%%%%%%%%%%%%%%%%%%%%%%%%%%%%   TABLE
%\begin{table} [ht]
%\vspace{-2mm}
%\centering
%\small
%\begin{tabular}{|c | c | c | c | c | c |}
%\hline
%  &$S$ & $U$ & $T$ & $b$ & $c$\\
% \hline
%$S$ & $F_1$ & $F_2$ & $F_{3}$ & $T_1$ & $T_2$\\
%\hline
%$U$ &  & $F_4$ & $F_5$ & $T_3$ & $F_6$\\
%\hline
%$T$ & &  & $F_7$ & $F_8$ & $F_9$\\
%\hline
%$b$ &  &  &  & $F_{10}$ & $F_{11}$\\
%\hline
%$c$ &  &  & &  & $F_{12}$\\
%\hline
% \end{tabular}

%\caption { The possibility of the pairs that can be covered after the candidate edge addition: F and T denote %impossibility and possibility respectively.}\label{tab:table_nphard}
%\vspace{-2mm}
%\end{table}

\end{document}